\documentclass{article}

\usepackage{arxiv}

\usepackage[utf8]{inputenc} %
\usepackage[T1]{fontenc}    %
\usepackage{hyperref}       %
\usepackage{url}            %
\usepackage{amsfonts}       %
\usepackage{nicefrac}       %
\usepackage{microtype}      %
\usepackage{graphicx}
\usepackage[numbers]{natbib}
\usepackage{doi}

\usepackage{amsthm}
\usepackage{algpseudocode}
\usepackage{algorithm}
\usepackage{tikz}
\usepackage{multirow}
\usepackage{amsmath}
\usepackage{array}
\usepackage{hyperref}
\usepackage{tabularx}
\usepackage{nicematrix}
\usepackage{subcaption}
\usepackage{adjustbox}
\usepackage{makecell}
\usepackage{todonotes}
\usepackage{bm}
\usepackage{arydshln}
\usepackage{booktabs}

\usepackage{amssymb}
\usepackage[most]{tcolorbox}

\usepackage{enumitem} %

\newcommand{\X}{\mathbf{X}}
\newcommand{\XX}{{\X_M}}

\newcommand{\UNP}{\emph{UNP}}
\newcommand{\UNPB}{\emph{UNPB}}

\def\namedlabel#1#2{\begingroup
    #2%
    \def\@currentlabel{#2}%
    \phantomsection\label{#1}\endgroup
} %

\DeclareMathOperator*{\argmax}{argmax}

\newtheorem{remark}{Remark}
\newtheorem{theorem}{Theorem}
\newtheorem{lemma}{Lemma}[section]
\newtheorem{definition}[lemma]{Definition}
\newtheorem{claim}[lemma]{Claim}
\newtheorem{observation}[lemma]{Observation}

\title{On the existence of EFX allocations in multigraphs}

\author{ 
	\href{https://orcid.org/0000-0003-3997-5131}{\includegraphics[scale=0.06]{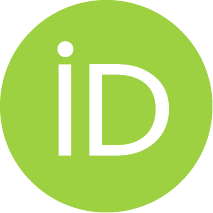}\hspace{1mm}Alkmini Sgouritsa} \\
	Athens University of Economics and Business\\
	Archimedes/Athena RC\\
	Greece \\
    \AND
	\href{https://orcid.org/0009-0007-3016-980X}{\includegraphics[scale=0.06]{orcid.pdf} \hspace{1mm}Minas Marios Sotiriou} \\
    National and Kapodistrian University of Athens\\
    National Technical University of Athens\\
	Athens University of Economics and Business\\
	Greece \\
}

\date{}

\renewcommand{\headeright}{}
\renewcommand{\undertitle}{}
\renewcommand{\shorttitle}{}

\hypersetup{
pdftitle={Title},
pdfsubject={},
pdfauthor={},
pdfkeywords={},
}

\begin{document}
\maketitle

\begin{abstract}
   We study the problem of ``fairly'' dividing indivisible goods to several agents that have valuation set functions over the sets of goods. As fair we consider the allocations that are envy-free up to any good (EFX), i.e., no agent envies any proper subset of the goods given to any other agent. The existence or not of EFX allocations is a major open problem in Fair Division, and there are only positive results for special cases. 

   Christodoulou et al.\cite{EFXsimplegraphs} introduced a restriction on the agents' valuations according to a graph structure: the vertices correspond to agents and the edges to goods, and each vertex/agent has zero marginal value (or in other words, they are indifferent) for the edges/goods that are not adjacent to them. The existence of EFX allocations has been shown for simple graphs with general monotone valuations \cite{EFXsimplegraphs}, and for multigraphs for restricted additive valuations \cite{kaviani2024envyfreeallocationindivisiblegoods}. 
   
   In this work, we push the state-of-the-art further, and show that the EFX allocations always exists in {\em multigraphs} and {\em general monotone valuations} if any of the following three conditions hold: either (a) the multigraph is bipartite, or (b) each agent has at most $\lceil \frac{n}{4} \rceil -1$ neighbors, where $n$ is the total number of agents, or (c) the shortest cycle with non-parallel edges has length at least 6.
\end{abstract}

\newpage

\section{Introduction}
We study a problem of ``fairly'' dividing indivisible goods to many agents. The question of how to divide resources to several agents in a fair way dates back to the ancient times, e.g., dividing land, and it raised important research questions since the late 40's \cite{h__steihaus_1948}. One prominent notion in fair division is {\em envy-free allocations}, where nobody envies what is allocated to any other agent, which was formally introduced a bit later  \cite{gamow1958puzzle, foley1966resource, VARIAN197463}. Initially, the problem was studied under the scope of divisible resources, where envy-free allocations are known to always exist \cite{Stromquist1980HowTC, Woo80, Aziz2016}.

The focus of this work is on indivisible goods, with multiple applications, such as dividing inheritance, and assigning courses to students \cite{courses}. The non-proﬁt website Spliddit (\url{http://www.spliddit.org/}) provides mechanisms for several such applications.  
It is easy to see that envy-free allocations are not guaranteed to exist; for instance consider two agents and one indivisible good, then whoever gets the good is envied by the other agent. 
This example demonstrates how strong the requirement of completely envy-freeness is for the scenario of indivisible goods. 

This has led to the study of two basic relaxations of envy-freeness, namely envy-freeness up to one good (EF1) \cite{budish} and envy-freeness up to any good (EFX) \cite{EFXCara}. EF1 is a weaker notion than EFX, and it is guaranteed to always exist and can be found in polynomial time \cite{LiptonEtAl}. On the other hand, it is not known if EFX allocations are guaranteed to exist in general, and it has been characterized as "Fair Division's Most Enigmatic Question" \cite{Procaccia}. EFX allocations are known to exist for special cases: 
e.g., for 2 agents with general monotone valuations \cite{PlautRough}, for 3 agents with additive valuations or a slightly more general class  \cite{CGM24, akrami2022efxallocationssimplificationsimprovements}, and for many agents with identical monotone valuations \cite{PlautRough}, or with additive valuations where each agent is restricted to have one of the two fixed values for each good \cite{TwoValuedInstanses}. 

Surprisingly, it was recently shown that EFX allocations need not exists in the case with chores, i.e., negatively valued items \cite{christoforidis2024pursuitefxchoresnonexistence2approx}. This is the first result of non-existence of EFX for monotone valuation functions, and the construction requires only 3 agents and 6 goods. This is an interesting separation between goods and chores, as for the case of goods it is known that EFX allocations are guaranteed to exist when the number of goods are at most 3 more than the number of agents \cite{mahara2021extensionadditivevaluationsgeneral}. 

Unfortunately, little is known for the case with multiple agents and multiple goods; additionally to the works that have been already mentioned \cite{TwoValuedInstanses, PlautRough}, EFX allocations are known to exist when agents' preference follow a lexicographic order defined by their preference over singletons \cite{DBLP:conf/aaai/HosseiniSVX21}, and when the valuations have dichotomous marginals, i.e., the marginal value when a good is added to a set is either 0 or 1 \cite{babaioff2020fairtruthfulmechanismsdichotomous}. All those works consider high restrictions and resemblance on the agents valuations. Towards broadening our understanding for the case of multiple agents and goods, Christodoulou et al.~\cite{EFXsimplegraphs} introduced a setting that is related to our work, where the valuations are defined based on a graph: given a graph, the agents correspond to the vertices of the graph, and the goods to the edges. Then, each agent is indifferent for the goods/edges that are not adjacent to them. In \cite{EFXsimplegraphs}, they showed that EFX allocations always exist on graphs.

In this work we consider a multigraph, which can be interpreted as follows: {\em each good is of interest for at most two agents}. The motivation in the multigraph setting is, similarly to \cite{EFXsimplegraphs}, the division of territories between nations, areas of interest between neighboring countries and more generally division of geographic settings. 
Another application is to allocate available space for research teams and collaborators, which is always a challenging task, and becomes even more difficult when there are multiple conflicts for available areas.
It was recently showed that EFX allocation always exists in multigraphs when all agents have restricted additive valuations \cite{kaviani2024envyfreeallocationindivisiblegoods}, i.e., each good $g$ has a fixed value $v_g$, and each agent may value the good by $v_g$ or not value it at all, in which case he has value $0$. We generalize this result with respect to the valuation functions, where we use general monotone valuations, on the expense of having restrictions on the multigraph, where either the multigraph is bipartite, or each agent has at most $\lceil \frac{n}{4} \rceil-1$ neighbors, where $n$ is the total number of agents, or the shortest cycle with no parallel edges has length at least 6. 

Two other papers independently and in parallel showed EFX allocations involving multigraphs \cite{afshinmehr2024efxallocationsorientationsbipartite, bhaskar2024efxallocationsmultigraphclasses}. We next discuss their results and the comparison to ours. \citet{afshinmehr2024efxallocationsorientationsbipartite}  independently showed that EFX allocations exist for bipartite multigraphs, which coincides with our first result. They further showed the existence of EFX allocations for multicycles with additive valuations. Our first and third results (for bipartite multigraphs and graphs with cycles of non-parallel edges of length at least $6$) include all multicycles apart from the ones of length $3$ and $5$; we note however that we consider the more general monotone valuations.    \citet{bhaskar2024efxallocationsmultigraphclasses} independently showed the existence of EFX allocations in bipartite multigraphs when agents have cancelable valuations and in multi-trees for general monotone valuations. Our first result generalizes those two results. Moreover, they showed that multigraphs with chromatic number $t$ admit EFX allocations when the shortest cycle using non-parallel edges is of length at least $2t-1$; this result holds when agents has cancelable valuations. Note that the multigraphs with chromatic number $2$ is just the bipartite multigraphs. When the chromatic number is at least $4$, we improve the requirement for shortest cycles using non-parallel edges, from length at least $7$ (and greater depending on the chromatic number of the multigraph) to $6$ (no matter the chromatic number), and furthermore our results is more general as it applies to general monotone valuations. On the other hand, for multigraphs with chromatic number 3, \citet{bhaskar2024efxallocationsmultigraphclasses} showed the existence of EFX allocation (for cancelable valuations) for an improved length of $5$ (comparing to $6$ in our result) for the shortest cycle using non-parallel edges.

\subsection{Our Results}

We show that an EFX allocation exists for multigraphs when agents have general monotone valuations and:
\begin{itemize}
    \item The graph is bipartite
    \item Each agent has at most $\lceil \frac{n}{4} \rceil-1$ neighbors, where $n$ is the total number of agents
    \item The shortest cycle with non-parallel edges of the multigraph is at least 6
\end{itemize}

Our results are summarized in the following three theorems: 

\begin{theorem}\label{theorem:bipartite}
    In bipartite multigraphs, an EFX allocation always exists.
\end{theorem}

\begin{theorem}\label{result2}
    In multigraphs with at most $\lceil \frac{n}{4} \rceil-1$ neighbors per agent, where $n$ is the total number of agents, an EFX allocation always exists. 
\end{theorem}

\begin{theorem}\label{result2Girth}
    In multigraphs, where the shortest cycle with non-parallel edges has length at least 6, an EFX allocation always exists.
\end{theorem}

The construction of EFX allocations follows the same skeleton for all our results, and so they are presented together. In order to keep the presentation smooth, we begin in Section~\ref{sec:2parallel} with a special case of multigraphs for which each pair of vertices is connected with {\em at most two} edges. As we discuss in Section~\ref{sec:Techniques}, in the general case, we consider different ways to partition the common edges between two vertices. For simplicity of the presentation, we postpone this crucial technicality for later and consider at most two parallel edges, for which there is a single way to be partitioned. This way we focus on highlighting first all the important ideas for constructing EFX allocations. In Section~\ref{sec:generalHighLevel} we consider the general case where each pair of vertices is connected with multiple edges.

\subsection{Our Techniques}
\label{sec:Techniques}
Here we discuss our main techniques in order to construct EFX allocations in multigraphs. 

We make use of the cut-and-choose-based protocol of \cite{PlautRough} for two agents: one agent cuts the set of goods into two bundles where he is EFX-satisfied with each of them (i.e., no matter which of the two bundles he receives, he does not envy the other bundle up to any good), and the other agent chooses his favorite bundle among those two. This simple protocol results in an EFX allocation for two agents, even if they have general monotone valuations. 
Note that there may be two different EFX allocations derived by the cut-and-choose protocol, depending on who ``cuts'', and moreover, only the agent who ``cuts'' might be envious of the other agent (i.e., the agent who chooses). 

\begin{remark}
  {\em We remark that according to the original definition of EFX in \cite{EFXCara}, where each agent $i$ is not envious against any other agent $j$ after the hypothetical removal of a {\em positive valued} good for $i$ from the $j$'s bundle, the cut-and-choose protocol provides a simple EFX allocation for multigraphs: for every pair of vertices $i,j$ with common adjacent edges $E_{ij}$, use the cut-and-choose protocol to decide the allocation of $E_{ij}$. This allocation is in fact an {\em orientation}, i.e., each edge is given to one of its endpoints, and satisfies EFX. The reason is that $i$ is indifferent about any other good that $j$ receives apart from $E_{ij}$, and $E_{ij}$ is allocated between $i$ and $j$ in such a way that $i$ is envy free up to any positively valued good against $j$. However, following the traditional definition of EFX, the good that is hypothetically
  removed from $j$'s bundle may be indifferent for $i$, and the local allocation of the cut-and-choose protocol is insufficient. Instead, we need to consider the whole multigraph more globally.}
\end{remark}

Following the above remark, we make use of the cut-and-choose protocol in order to partition the set $E_{ij}$ into two bundles, but we may consider two different partitions depending on which endpoint ``cuts''. The reason for that is so that we control the direction of envy, and manage to generalize the ideas of \cite{EFXsimplegraphs}. However, we also put the cut-and-choose protocol in use in a different way: if two agents do not agree on having the same cut, we use the EFX-cut of one of them in order to create a partition of {\em three} bundles where the two agents have different most valued set. This tool was proven to be very useful for constructing the EFX allocation for Theorem~\ref{result2}, where we want to minimize the number of envied agents. In both approaches, one crucial condition that we always upkeep is that {\em we never allocate more than one bundle of the partition of $E_{ij}$ to the same vertex. }

Our approach can be seen as a three-step procedure: i) We define an initial allocation where each agent receives exactly one bundle (derived from carefully constructed partitions) from the common edges with exactly one of his neighbors. In this step we guarantee some ground properties on the allocation. %
ii) We perform Algorithm~\ref{Algo1(2,2)} (a generalization of Algorithm 2 of \cite{EFXsimplegraphs}) in order to satisfy extra properties by preserving an EFX orientation, while ensuring that any non-envied agent has received exactly one adjacent bundle associated with {\em each} of their neighbors. At this step we have finalized any orientation of the edges, whose allocation will not change in the next step. iii) We appropriately allocate all the unallocated edges to non-envied vertices that are {\em not endpoints} of the edges, while preserving the allocation of Step 2 and the EFX guarantee.

\paragraph{Initial Allocation.} One important novelty of our work is about the construction of the initial allocation in Step 1. In all our three results, we carefully define for each pair of vertices a partition of their common edges into 2 or 3 bundles, such that in the initial allocation the following two properties are satisfied: 

\begin{enumerate}
    \item it is an EFX orientation; 
    \item nobody prefers an unallocated bundle of edges (based on a defined partition) to what they have.
\end{enumerate}
For showing Theorem~\ref{theorem:bipartite}, where the multigraph is bipartite, we use the cut-and-choose protocol for partitioning the edges, where the vertices on the one side EFX-cut and the vertices from the other side choose. Then, the vertices from the former side choose their favorite available bundle. This results in an initial allocation where only vertices from the latter side may be envied, which provides the following property:
\begin{center}
{\em No two envied vertices are adjacent.} 
\end{center}

For Theorem~\ref{result2}, we construct an initial allocation with bounded envy as follows: we properly partition the parallel edges between any pair of vertices into 2 or 3 bundles; if there is a partition into two bundles that is a ``cut'' for both endpoints we keep this partition, otherwise we construct a partition into 3 bundles where the top preference for each endpoint is different (Lemma~\ref{lemma:notsamepartition}). Given those partitions, we create a special weighted bipartite graph between all vertices and the created bundles of edges, where each vertex is only connected with its most preferred and second most preferred bundle with appropriate weights. Then, the maximum matching in that graph induces an initial allocation that satisfies the following property:
\begin{center}
{\em At most half of the agents are envied.} 
\end{center}

For Theorem~\ref{result2Girth}, starting by any partial allocation that satisfies the above Properties 1 and 2, we alter the allocation in order to further ensure the following property: 
\begin{center}
{\em For each envied vertex there exists a non-envied vertex in short distance (at most 2).}
\end{center}

{\bf Reducing envy from the initial allocation.} In step 2, we run Algorithm~\ref{Algo1(2,2)}, by focusing on reducing the number of envied vertices and giving as many bundles as possible to non-envied vertices. The initial allocation changes that way, laying the groundwork for the allocation of the remaining edges.

{\bf Final allocation.} In step 3, we have the necessary conditions to allocate the rest of the unallocated bundles. We assign each such bundle to a non-envied vertex that is not an endpoint to that bundle.

\subsection{Further Related Work}

We focus on references related to EFX, and we defer the reader to a recent survey \cite{DBLP:journals/ai/AmanatidisABFLMVW23} that discusses other notions of fairness, as well.

In \cite{EFXsimplegraphs}, they introduced the $(p,q)$-bounded setting, where an item is relevant to at most $p$ agents, and has multiplicity $q$. In \cite{kaviani2024envyfreeallocationindivisiblegoods}, they use a bid altered notation where $q$ represents the maximum number of items that are relevant to a pair of agents. In both notations, the multigraph setting is equivalent to the $(2,\infty)$-bounded setting. The existence of EFX allocations in the $(2,1)$-bounded setting, which is equivalent to simple graphs, has been studied for goods \cite{EFXsimplegraphs}, and for mixed manna settings \cite{mixedmanna}.  
In the graph and multigraph setting the existence of EFX orientations, i.e., allocations where edges may only be allocated to one of the endpoints, has also been considered. In \cite{EFXsimplegraphs} they showed that EFX orientations need not exist by giving a counterexample in a $K_4$ graph, and they further showed that even deciding if there exists an EFX orientation is NP-complete; it was later shown that this result holds even if the vertex cover of the graph has size of 8, or in multigraphs with only 10 vertices \cite{deligkas2024ef1efxorientations}. In \cite{OnTheStructureOfEFXOrientationsOnGraphs}, they showed that EFX orientations are not guaranteed to exist in graphs with chromatic number greater than 3, and they always exist when the chromatic number is at most 2. Regarding multigraphs, a very recent work  \cite{hsu2024efxorientationsmultigraphs} showed that finding an EFX orientation is NP complete even for bipartite multigraphs.

Several relaxations of EFX have also been explored. One such relaxation is the approximate EFX, $\alpha$-EFX, where no agent envies any strict subset of another agent's allocated set by more than an $\alpha \leq 1$ parameter. For subadditive valuations a $\frac{1}{2}$-EFX allocation is known \cite{PlautRough}. The approximation was improved to $\phi - 1$, but for the more restricted class of additive valuations \cite{Amanatidis_Markakis_Ntokos_2020}. The approximation has been pushed further to $\frac{2}{3}$ for additive valuations under several restrictions \cite{amanatidis2024pushingfrontierapproximateefx, OrdinalAssumptions}. Related to our work, in \cite{amanatidis2024pushingfrontierapproximateefx} they showed $\frac{2}{3}$-EFX for additive valuations in the setting of multigraphs. For the $(\infty,1)$ - bounded setting (according to \cite{kaviani2024envyfreeallocationindivisiblegoods}), $\frac{\sqrt{2}}{2}$-EFX allocations are guaranteed to exist even for subadditive valuations \cite{kaviani2024envyfreeallocationindivisiblegoods}.

 Another relaxation of EFX, introduced by Caragiannis et al. \cite{CGH19}, is to consider partial EFX allocations, also known as EFX with charity. Chaudhury et al. \cite{CKMS21} showed that EFX allocations always exist, even with general monotone valuations, if at most $n-1$ goods are donated to charity, where $n$ is the number of agents, and moreover nobody envies the charity. The size of the charity was then improved to $n-2$, and to one for the case of 4 agents with additive valuations \cite{AlmostFullEFXforFourAgents}. For the $(\infty,1)$ - bounded setting  (according to \cite{kaviani2024envyfreeallocationindivisiblegoods}), the size of the charity was reduced to $\lfloor \frac{n}{2} \rfloor -1$ for general monotone valuations \cite{kaviani2024envyfreeallocationindivisiblegoods}. Finally, the number of unallocated goods was subsequently improved to sublinear by \cite{CGMMM21, akrami2022efxallocationssimplificationsimprovements, berendsohn2022fixedpointcyclesefxallocations, jahan2023rainbowcyclenumberefx}, but for $(1-\varepsilon)$-EFX, for any $\varepsilon \in (0,\frac{1}{2}]$.

\section{Preliminaries}
\label{sec:prel}
We consider a setting where there is a set $N$ of $n$ agents and a set $M$ of $m$ indivisible goods, and each agent $i$ has a valuation set function $v_i: 2^M \rightarrow \mathbb{R}$, over the sets of goods, i.e., by $v_i(S)$ we denote the valuation of agent $i$ for the set $S$ of goods. The valuation functions are considered to be monotone, i.e.,  for any $S\subseteq T \subseteq M$, it holds that $v_i(S) \leq v_i(T)$, and normalized, i.e., $v_i(\emptyset)=0$. 
We slightly abuse notation, so that if $S$ is a set of bundles, $v(S)=v(\cup_{B\in S} B)$.
The only restriction that is considered about the valuations is described by the following multigraph setting. 

{\bf Multigraph setting.} Let $G=(V,E)$ be a multigraph with $|V| = n$ vertices  and $|E| = m$ edges. In the multigraph setting, the agents correspond to the vertices of $G$, and the goods correspond to the edges of $G$. From now on we will refer to the agents as vertices and to goods as edges. In this setting, the edges that are not adjacent to some vertex $i$ are {\em irrelevant} to it, i.e., for any $S\subseteq E$ and any $g \in E$ that is not adjacent to $i$, $v_i(S \cup \{g\}) = v_i(S)$. We call a good {\em relevant} to $i$, if it is adjacent to it. We further call any two edges with the same endpoints  {\em parallel}.

{\bf Allocation.}
An {\em allocation} $\X=(X_1, \ldots, X_n)$ is a partition of a subset of $E$ into $n$ (disjoint) bundles $X_1, \ldots, X_n$, where each vertex $i$ receives $X_i$. We call an allocation \emph{complete} if it's a partition of the set $E$. We call an allocation \emph{partial} if it's a partition of a strict subset of $E$.  

{\bf Social Welfare.}
Given an allocation $\X$ to $n$ agents, the \emph{social welfare} is defined as $\sum^n_{i=1} v_i(X_i)$.

{\bf Orientation.}
An {\em orientation} is an allocation where each allocated edge is given to one of its endpoints.

{\bf Unallocated edges.}
    Given a partial allocation $\X$, an edge $e$ is \emph{unallocated} if for every $X_i$, $e\notin X_i$. We denote by $U(\X)$ the set of unallocated edges in $\X$. For each vertex $i$, we define $U_i(\X)$ to be the set of all the unallocated edges that are relevant to $i$.

{\bf Envy - EFX-satisfied.}
    Given an allocation $\X=(X_1,\ldots, X_n)$, we say that a vertex $i$ {\em envies} another vertex $j$ (or alternatively $X_j$), if
    $v_i(X_i) < v_i(X_j)$. 
    We say that $i$ is {\em EFX-satisfied} against $j$ (or $X_j$) given $\X$, if
    $v_i(X_i) \geq v_i(X_j\setminus\{g\})$, for all $g\in X_j$.

{\bf EFX allocation.}
    An allocation $\X=(X_1,\ldots, X_n)$ is {\em EFX} if for any pair of vertices $i,j$ it holds that:
    $$v_i(X_i) \geq v_i(X_j \setminus \{g\}) \hspace{1ex} ,\forall g \in X_j.$$

{\bf EFX-cut.} We heavily rely on the cut-and-choose protocol \cite{PlautRough} when considering the parallel edges $E_{ij}$ between two vertices $i$ and $j$: $i$ ``cuts'' $E_{ij}$ by partitioning it into two bundles such that whichever bundle of the two $i$ receives, it is EFX-satisfied against the remaining bundle. We call such a partition an \emph{EFX-cut} for $i$ with respect to vertex $j$ (we may omit referring to $j$ if it is clear from the context). In the the cut-and-choose protocol \cite{PlautRough}, then $j$ receives its most valued bundle and $i$ the remaining bundle.

{\bf At least degree 2 in $G$.} W.l.o.g., we assume that each vertex has degree {\em at least $2$}, otherwise, if some vertex $i$ has degree $1$, meaning that vertex $i$ is interested only in one good $g$, we may allocate only $g$ to $i$, resulting in $i$ not envying anybody else, and the condition of EFX is satisfied for any other vertex against $i$ (no matter what they receive), since $i$ is allocated a single good.

\begin{observation}\label{observation:orientation}
    Given an EFX orientation $\X$, if some vertex $i$ is envied, then there exists a single vertex $j$ that envies $i$, and $X_i$ contains only edges relevant to $j$ (all edges in $X_i$ are relevant to $j$).
\end{observation}

\begin{proof}
Suppose on the contrary that vertex $j$ envies vertex $i$, and there exists edge $g \in X_i$ such that $g$ is irrelevant to $j$. Then $v_j(X_i \setminus g) = v_j(X_i) > v_j(X_j)$, where the equality is by definition of irrelevant edges, and the inequality is true because $j$ envies $i$. This is a contradiction to the fact that $\X$ satisfies EFX. Therefore, all edges in $X_i$ are relevant to $j$, and to $i$ since $\X$ is an orientation. This in turn means that for any other vertex all edges in $X_i$ are irrelevant, and therefore there is no other vertex envying $i$.
\end{proof}

\section{At Most 2 Parallel Edges}
\label{sec:2parallel}
In this section we prove Theorems ~\ref{theorem:bipartite}, ~\ref{result2} and~\ref{result2Girth} for the case where each pair of vertices are connected with {\em at most two} edges. Regarding Theorem~\ref{result2}, in this special case of at most 2 parallel edges, we manage to relax slightly the restriction to at most $\lfloor \frac{n}{4} \rfloor$ neighbors (instead of at most $\lceil \frac{n}{4} \rceil -1$). We next give some further definitions regarding this special case.

{\bf Edges $e^1_{ij}, e^2_{ij}$.} For any pair of adjacent vertices $i,j$, we call $e^1_{ij}, e^2_{ij}$, the $i$'s most and second most valued edge between vertices $i,j$, respectively. If there is only one edge between $i$ and $j$, w.l.o.g., we may add a dummy edge as their least preferred edge. Moreover, we give the following two useful definitions.

\begin{definition}\label{def:p_i(X)}
    Given an EFX allocation $\X$, for any envied vertex $i$, we denote the vertex that envies them as $p_i(\X)$. 
\end{definition}

\begin{definition} 
\label{def:UNP}
    
    Given an EFX allocation $\X = (X_1,\dots,X_n)$, and an {\em envied} vertex $i$, let $\hat{U}_i(\X)=U_i(\X)\cup\{e^2_{ip_i(\X)}\}$. We define the most valued set of potentially unallocated non-parallel edges as
    $$\UNP_i(\X) \in  \arg\max_{S \subseteq \hat{U}_i(\X)}\{v_i(S)| \text{ any $e_1,e_2 \in S$ are not parallel}\}\,,$$ %
    where $\UNP_i(\X)$ is chosen to have the {\em maximum possible cardinality}.
    This is the best bundle of non-parallel edges that $i$ could get if $X_i$ would be given to $p_i(\X)$ (who envies $i$), and $X_{p_i(\X)}$ becomes available.
\end{definition}

\begin{definition} 
\label{def:UNPnon-envied}
    
    Given an EFX allocation $\X = (X_1,\dots,X_n)$, and a {\em non-envied} vertex $i$, let $\hat{U}_i(\X)=U_i(\X)\cup X_i$. We define the most valued set of potentially unallocated non-parallel edges as
    $$\UNP_i(\X) \in  \arg\max_{S \subseteq \hat{U}_i(\X)}\{v_i(S)| \text{ any $e_1,e_2 \in S$ are not parallel}\} \,.$$ %
    where $\UNP_i(\X)$ is chosen to have the {\em maximum possible cardinality}. This is the best bundles of non-parallel edges that $i$ could get without changing the allocation of the other vertices.
\end{definition}

\begin{remark}
\label{rem:for_prop4}
    {\em Given an EFX allocation $\X$, suppose that some vertex $i$ is assigned $\UNP_i(\X)$, and the allocation of the other vertices remain unchanged, apart maybe from 
    vertex $p_i(\X)$ in the case that $i$ is envied; if $p_i(\X)$ received $e^2_{ip_i(\X)}$ in $\X$, and that edge belongs to $\UNP_i(\X)$, then it is removed from $p_i(\X)$'s bundle and is given to $i$ along with the $\UNP_i(\X)$. In the new allocation, there are no parallel edges adjacent to $i$ that are both unallocated. The reason is that $\UNP_i(\X)$ has maximum value and cardinality, so one of the two parallel edges should belong to $\UNP_i(\X)$.}
\end{remark}

We prove Theorems~\ref{theorem:bipartite},~\ref{result2} and~\ref{result2Girth} by constructions of an EFX allocation in three steps. The first two steps are dedicated to satisfy specific properties for the allocation that is always an EFX orientation, and the final step is an assignment of the remaining unallocated edges to vertices different from their endpoints. We discuss each step separately, and we first give the required properties for the first two steps. 

\begin{tcolorbox}[colback=black!5!white,colframe=black!75!black]
Properties of the allocation $\X$ after the 1st step.
    
\begin{itemize}
    \item[(1)] A partial EFX orientation.
    \item[(2)] For any vertex $i$ and $e \in U(\X), v_i(X_i) \geq v_i(e)$.
    \item[(3.1)] No two envied vertices are adjacent. (For Theorem~\ref{theorem:bipartite})\label{Property: Bipartite}
    \item[(3.2)] The number of the envied vertices is at most $\lfloor \frac n2\rfloor$. (For Theorem~\ref{result2}) \label{Property: BoundedEnvied}
    \item[{(3.3)}]For any envied vertex $i$, $p_i(\X)$ or at least one neighbor of $p_i(\X)$ is non-envied. (For Theorem~\ref{result2Girth})\label{Property: ForGirth}
\end{itemize}
Additional property of the allocation $\X$  after the 2nd step.
\begin{itemize}
    \item[{(4)}] For any envied vertex $i$, $v_i(X_i) \geq v_i\left(\UNP_i(\X)\right),$ and any non-envied vertex $j, v_j(X_j) \geq v_j(U_j(\X)).$ 
    \label{Property:NoEnvyNonParallel} 
\end{itemize}
\end{tcolorbox}

\subsection{Step 1 - Initial Allocations}
This step is the initial allocation towards the EFX construction. For each of our results, in this warm up case, i.e., at most 2 parallel edges between any pair of vertices, we show an initial allocation that satisfies Properties (1),(2) and respectively one of the Properties (3.1),(3.2) and (3.3) for each of our main theorems.

\subsubsection{Initial allocation for bipartite multigraphs}
Here we describe the initial allocation for the case of some bipartite multigraph $G=(V=A\cup B,E)$: starting from the empty allocation, each vertex in side $A$ is allocated their most valued edge. Then each vertex in side $B$ is allocated their most valued edge from the remaining unallocated edges. This procedure is formally described in Algorithm~\ref{Algo:InitBipar}.

\begin{algorithm}
\caption{Initial Allocation for Bipartite Graphs}
\label{Algo:InitBipar}
\raggedright\textbf{Input:} A bipartite multigraph $G = (V = A \cup B,E)$. \\
\textbf{Output:} $\X$ satisfying Properties (1), (2), and (3.1).
\begin{algorithmic}[1]
\For{every vertex $i$}
    \State $X_i \gets \emptyset$
\EndFor
\For{every $i \in A$}
    \State $X_i \gets e_i \in \arg \max_{e \in U_i(\X)}\{v_i(e)\}$
\EndFor

\For{every $j \in B$}
    \State $X_j \gets e_j\in \arg \max_{e \in U_j(\X)}\{v_j(e)\}$
\EndFor
\end{algorithmic}
\end{algorithm}

\begin{lemma}
    Algorithm~\ref{Algo:InitBipar} outputs an allocation $\X$ that satisfies Properties (1), (2), and (3.1).
\end{lemma}

\begin{proof}
     Property (1) is trivially satisfied since in $\X$ every vertex receives a single edge adjacent to them. Property (2) is also satisfied since when Algorithm~\ref{Algo:InitBipar} updates the bundle allocated to any vertex $i$, they receive their most valued edge, so they prefer $X_i$ to any unallocated edge. 
          For Property (3.1) note that all vertices in side $A$ are assigned their most valued edge before assigning any edge to side $B$. Since the graph is bipartite, only vertices from side $B$ may envy vertices from side $A$. Therefore, all envied vertices are on side $A$, and in turn, no two envied vertices are adjacent.
\end{proof}

\subsubsection{Initial allocation for at most $\lfloor \frac{n}{4} \rfloor$ neighbors per vertex}

The purpose in this case is to find an initial allocation with bounded number of envied vertices. For this, we carefully define a bipartite (simple) weighted graph $H(G) = (A \cup B, E_H)$ based on the given multigraph $G$, between the vertices and the edges of $G$. A maximum weighted matching of $H(G)$ gives an allocation where at least half of the vertices receive their most valued edge, and the rest receive their second most valued edge (so each of them envies at most one other vertex), limiting that way the number of envied vertices. 

\begin{definition}\label{def:bipartite}(Bipartite graph $H(G)$).
     \textnormal{Given a multigraph $G=(V,E)$, we define a simple weighted bipartite graph $H(G) = (A \cup B,E_H)$, where $A=V$ and $B=E$. %
    We create the edge set $E_H$ as follows: we connect each vertex $i$ of $A$ (i.e., each vertex of $G$) with the {\em two} vertices of $B$ that represents $i$'s most and second most valued edge in $G$ (adjacent to $i$) \footnote{Note that in Section~\ref{sec:prel} we argued that w.l.o.g., each vertex in $G$ is assumed to have degree at least $2$.}; we set the weights of those two edges we add in $H(G)$ to $1$ and $0$, respectively.} 
\end{definition}

We prove in Lemma~\ref{2,2match} that the allocation, let it be $\XX$, that corresponds to the maximum $A$-perfect matching\footnote{An $A$-perfect matching is a matching where all vertices of side $A$ are matched.} in $H(G)$ satisfies Properties (1),(2), and (3.2). We first show in Observation~\ref{obs:existence_PerfectMatching} that there is always an $A$-perfect matching.

\begin{observation}
\label{obs:existence_PerfectMatching}
    There is always an $A$-perfect matching in $H(G)$.
\end{observation}

\begin{proof}
    In order to prove the existence of an $A$-perfect matching, we show that Hall's theorem \cite{Hall1935OnRO} holds for any $S \subseteq A$. Let $N_S$ be the set of neighbors of $S$. For the sake of contradiction, assume that $|N_S|< |S|$ for some $S\subseteq A$. Each node in $S$ has degree exactly $2$, by construction, so there are exactly $2|S|$ edges in $H(G)$ with endpoints in $S$. Since we assumed $|N_S|< |S|$, by the pigeonhole principle there is a vertex $j \in N_S$ with degree at least $3$. However, each vertex in $B$, and therefore in $N_S$, has degree at most two, as it represents an edge in $G$ and it may only be connected with its endpoint vertices in $A$. This is a contradiction to our assumption, so for any $S\subseteq A$, it holds that $|N_S|\geq |S|$; by Hall's theorem, there is always an $A$-perfect matching in $H(G)$. 
\end{proof}

\begin{lemma}\label{2,2match}
    $\XX$ satisfies Properties (1),(2), and (3.2).
\end{lemma}

\begin{proof}

     Let $M$ be the maximum weighted $A$-perfect matching that corresponds to the allocation $\XX$. Every vertex is assigned exactly one edge that is adjacent to it, so Property (1) trivially holds. Regarding Property (2), note that every vertex gets either their first or second most preferred edge in the graph. If a vertex gets their most preferred edge, then Property (2) trivially holds for that vertex. If a vertex gets their second most preferred edge, then since $\XX$ corresponds to a {\em maximum} $A$-perfect matching, their most preferred edge is allocated to another vertex, so Property (2) holds for that vertex as well.

    Let $K$ be a maximal alternating path or an alternating cycle of $M$, including $k$ vertices of $A$. 
    Then, if we replace $K \cap M$ with $K \setminus M$ (so we alter the matching on $K$), we end up with another $A$-perfect matching. If $K \setminus M$ had greater weight than $ K \cap M$ then we would create a matching with higher weight which is a contradiction. So in every alternating cycle or path $K$, the weight of $K \cap M$ is at least $ \left \lceil \frac{k}{2} \right \rceil $. Since $M$ is an $A$-perfect matching, each vertex in $A$ belongs to one alternating path or cycle. Moreover, since each vertex in $A\cup B$ has degree at most $2$, no two alternating paths or cycles intersect. Therefore, 
    adding up over all alternating paths and cycles, we get that $M$ has weight at least $\lceil \frac{n}{2} \rceil$, meaning that at least $\lceil \frac{n}{2} \rceil$ vertices receive their most preferred edge in $\XX$. This in turn means that there are at most $\lfloor \frac{n}{2} \rfloor$ vertices receiving in $\XX$ their second most preferred edge. Those vertices are the only vertices that envy some other vertex, and they may envy at most one other vertex (the one that receives their most valued edge). Therefore, there are at most $\lfloor \frac{n}{2} \rfloor$ envied vertices, and hence Property (3.2) is satisfied.
\end{proof}

\subsubsection{Initial allocation for length of the shortest cycle with non-parallel edges at least 6}

This initial allocation requires an allocation that satisfies Properties (1) and (2) and then we change that allocation such that it additionally satisfies Property (3.3). We may use the initial allocation $\XX$, however there is a very simple greedy algorithm that satisfies Properties (1) and (2): iteratively let each vertex pick their most valued edge from its adjacent unallocated edges. We present this procedure in Algorithm~\ref{Algo:InitSimpleGreedy} 

\begin{algorithm}
\caption{Initial Greedy Allocation}
\label{Algo:InitSimpleGreedy}
\raggedright\textbf{Input:} A multigraph $G = (V,E)$. \\
\textbf{Output:} $\X$ satisfying Properties (1) and (2). 
\begin{algorithmic}[1]
\For{every vertex $i$}
    \State $X_i \gets \emptyset$
\EndFor
\For{every vertex $i$}
    \State $X_i \gets e_i \in \arg \max_{e \in U_i(\X)}\{ v_i(e)\}$
    
\EndFor
\end{algorithmic}
\end{algorithm}

\begin{lemma}
    Algorithm~\ref{Algo:InitSimpleGreedy} satisfies Properties (1) and (2).
\end{lemma}
\begin{proof}
    Every vertex is allocated exactly one adjacent edge, therefore we have an EFX orientation, and Property (1) is satisfied. Each vertex gets (in its turn) their most valued adjacent unallocated edge. Since they are indifferent for non-adjacent edges, overall, after running Algorithm~\ref{Algo:InitSimpleGreedy}, they do not prefer any unallocated edge to their allocated edge, resulting in satisfaction of Property (2).
\end{proof}

{\bf Satisfying also Property (3.3).} Let $\X$ be an allocation satisfying Properties (1) and (2), e.g., the allocation derived after running Algorithm~\ref{Algo:InitSimpleGreedy}. We use Algorithm~\ref{Algo1(InitialThmGirth)} to transform $\X$ so that it additionally satisfies Property (3.3). 

The idea of this algorithm is that when for an envied vertex $i$, the vertex $j = p_i(\X)$ is also envied, and so are  all its neighbors, $j$ releases its allocated edge which is given to the vertex that envies $j$. Then, $j$ receives its most valued available edge, if any, and Properties (1) and (2) still hold. Note that $j$ becomes non-envied and all its neighbors were envied before, and they may or may not be envied after the reallocation; so, overall, the envy only reduces, and Property (3.3) is now satisfied for $i$. 

\begin{algorithm}
\caption{Initial Allocation for Theorem~\ref{result2Girth}}
\label{Algo1(InitialThmGirth)}
\raggedright\textbf{Input:} An allocation $\X$ satisfying Properties (1)-(2). \\
\textbf{Output:} An allocation satisfying Properties (1),(2),(3.3).
\begin{algorithmic}[1]
\While{$\exists$ envied vertex $i$ for which Prop. (3.3) is not satisfied} 
    \State Let $j=p_i(\X)$ and $k=p_j(\X)$, then %
    \State $X_k \gets X_j$
    \State $X_j \gets \emptyset$ %
    \While{$\exists$ vertex $l$ and $e \in U_l(\X)$ s.t. $v_l(e) > v_l(X_l)$ }
        \State $X_l \gets \{e\}$ 
    \EndWhile
\EndWhile
\end{algorithmic}
\end{algorithm}

\begin{lemma}
\label{lem:Algo1Girth}
Algorithm~\ref{Algo1(InitialThmGirth)} terminates and the allocation returned satisfies Properties (1),(2) and (3.3). 
\end{lemma}

\begin{proof}
    We will show that if before each round of Algorithm~\ref{Algo1(InitialThmGirth)} Properties (1) and (2) were satisfied, those properties hold after the completion of that round (outer ``while''), and moreover, that the number of envied vertices is reduced by at least one. This would mean that Algorithm~\ref{Algo1(InitialThmGirth)} terminates and Properties (1),(2),(3.3) are satisfied; Property (3.3) should be trivially satisfied after the termination of Algorithm~\ref{Algo1(InitialThmGirth)} due to the condition of the outer ``while''.

    It is easy to see that Properties (1) and (2) are satisfied after any round of the outer ``while'': Property (1) is satisfied since each vertex receives at most one of its adjacent edges, and Property (2) is guaranteed by the inner ``while''. 
    Next we show that in each outer ``while'', no non-envied vertex becomes envied, and at least one envied vertex, namely $j$, becomes unenvied; this means that Algorithm~\ref{Algo1(InitialThmGirth)} terminates.
    
    Suppose that at the beginning of the outer ``while'' where some vertex $i$ is considered, Properties (1) and (2) are satisfied, and vertex $j=p_i(\X)$ and all its neighbors are envied, otherwise Property (3.3) would be satisfied for $i$. 
    In lines 3-4, $X_j$ is given to $k$, so $k$ is still envied but the envy now comes from $j$. $j$ is not envied anymore and it may only envy its neighbors, however those were envied before. Regarding the inner ``while'', a vertex $l$ may only change its assignment, after its neighbor releases an edge that $l$ prefers, however, its neighbor released that edge because it received a better one; so, no new envy may appear. In the special case that vertex $j$ is considered for the first time and receives an unallocated edge, still nobody would envy her, since Property (2) was satisfied before. Overall, an extra envy may only come from vertex $j$ to its neighbors (e.g., $j$ envies $k$), who were envied before, and $j$ becomes unenvied. Hence, the lemma follows.
\end{proof}

\subsection{ Step 2 - Satisfying Property \hyperref[Property:NoEnvyNonParallel]{(4)}}
This step starts with the initial allocation constructed in Step 1 satisfying Properties (1),(2), and one of (3.1),(3.2), and (3.3). In Step 2, the initial allocation changes according to Algorithm~\ref{Algo1(2,2)} in order to additionally satisfy Property (4).
We remark that all our results use the same algorithm in Step 2, no matter which of the three cases is considered.

The necessity of the first part of Property (4) is so that for each {\em envied} vertex $i$ we can allocate its adjacent unallocated edges $U_i(\X)$. We cannot allocate any of those edges to $i$ as long as it is envied. So, the high level idea for Step 2 is that either $U_i(\X)$, or more accurately $\UNP_i(\X)$ (see Definition~\ref{def:UNP}), is valued for $i$ more than $X_i$, in which case $i$ receives this bundle and releases $X_i$ (which in turn is given to $p_i(\X)$), or $i$ doesn't value $\UNP_i(\X)$ that much, and we can find ways to allocate it to other non-envied vertices (in Step 3) without causing envy to $i$. 
The idea is similar to Algorithm 2 of \cite{EFXsimplegraphs} but adjusted in the multigraph setting where there are parallel edges. 

The necessity of the second part of Property (4) is so that for each {\em non-envied} vertex $i$ we can allocate its adjacent unallocated edges $U_i(\X)$. Vertex $i$ ``chooses'' its most valued available set of non-parallel edges, so that the remaining unallocated set may be assigned to other non-envied vertices (in Step 3) without causing envy to $i$. This issue did not appear in simple graphs \cite{EFXsimplegraphs}, but only in multigraphs, and therefore, this part is also essential in order to derive a complete EFX allocation.

\begin{algorithm}
\caption{Reducing Envy Algorithm}
\label{Algo1(2,2)}
\raggedright\textbf{Input:} An allocation $\X$ satisfying Properties (1) and (2). \\
\textbf{Output:} An allocation satisfying Properties (1), (2) and (4). The allocation also preserves any of the Properties (3.1),(3.2),(3.3) if they were initially satisfied.
\begin{algorithmic}[1] %
\While{ $\exists$ non-envied vertex $k$ s.t. $v_k(\UNP_k(\X))> v_k(X_k)$ or 
\Statex \hspace{15pt} ($v_k(\UNP_k(\X)) = v_k(X_k)$ and $|\UNP_k(\X)| > |X_k|)$}
    \State $X_k \gets \UNP_k(\X)$ 
\EndWhile

\While{$\exists$ envied vertex $i$ s.t. $v_i(\UNP_i(\X))> v_i(X_i)$}
    \If{$e^2_{ip_i(\X)}\in \UNP_i(\X)$} 
        \State $X_{p_i(\X)} \gets X_{p_i(\X)} \setminus \{e^2_{ip_i(\X)}\}$
    \EndIf
    \State $X_i \gets \UNP_i(\X)$ 
    \While{$\exists$ vertex $j$ and $e \in U(\X)$ s.t. $v_j(e) > v_j(X_j)$} 
        \State $X_j \gets \{e\}$ 
    \EndWhile
    \While{ $\exists$ non-envied vertex $k$ s.t. $v_k(\UNP_k(\X))> v_k(X_k)$ or 
\Statex \hspace{25pt} ($v_k(\UNP_k(\X)) = v_k(X_k)$ and $|\UNP_k(\X)| > |X_k|)$}
        \State $X_k \gets \UNP_k(\X)$ 
    \EndWhile
\EndWhile
\end{algorithmic}
\end{algorithm}

\begin{lemma}\label{terminationAlgo1(2,2)}
    Algorithm~\ref{Algo1(2,2)} terminates and outputs an allocation $\X$ that satisfies Properties (1), (2) and (4). It further preserves any of the Properties (3.1),(3.2) and (3.3), if they were satisfied before applying Algorithm~\ref{Algo1(2,2)}.
\end{lemma}

\begin{proof}
    We first argue that Properties (1) and (2) are satisfied after any outer ``while'', i.e., the ones at lines 1-3 and 4-15, if they were satisfied before. Meanwhile, we show that any non-envied vertex cannot be envied after Algorithm~\ref{Algo1(2,2)}. 
    
    For the first ``while'' (lines 1-3), a non-envied vertex $k$ gets a new bundle that it values at least as much as its previous bundle, which is composed of non-parallel edges among the unallocated edges and his own bundle (see Definition~\ref{def:UNPnon-envied}). For each vertex $\ell$ adjacent to $k$, the new bundle of $k$ contains a single edge relevant to $\ell$, that $\ell$ doesn't value more than its own bundle, due to Property (2) and to the fact that $k$ was not envied. So, $k$ remains non-envied. Moreover, by the definition of the $\UNP_k(\X)$ set, Properties (1) and (2) are satisfied for vertex $k$. Property (2) is also satisfied for the rest of the vertices since it was satisfied before and due to the fact that $k$ was not envied (so, any edges released by $k$ are not more valued to what the other vertices get). The same arguments hold for the ``while'' of lines 12-14, which is identical to the one in lines 1-3, if Properties (1) and (2) are satisfied before it is executed, which is what we will show next. 
    
    Regarding the ``while'' in lines 4-15, we will show that if Properties (1) and (2) were satisfied before each round of that ``while'', they are satisfied when the inner ``while'' at lines 9-11 terminates. In lines 4-8, an envied vertex $i$ receives a set that may be only envied by $p_i(\X)$, due to Property (2) and the definition of the set $\UNP_i(\X)$ (see Definition~\ref{def:UNP}). Moreover, Property (2) is satisfied for any other vertex apart from $p_i(\X)$, since the only edge that was released was an edge between $i$ and $p_i(\X)$ (see Observation~\ref{observation:orientation}). Therefore, the inner ``while'' in lines 9-11 will first run for $j=p_i(\X)$, where $p_i(\X)$ will receive the edge that $i$ released and caused the envy towards $i$ in the first place. Overall, in lines 4-11 the number of envied vertices reduces by at least 1 (namely vertex $i$ becomes non-envied), and moreover in the ``while'' in lines 9-11 only an envied vertex may become non-envied and not vice versa, since a vertex may only change its assignment, after its neighbor releases an edge while receiving a better one. Additionally, after the ``while'' in lines 9-11, Property (1) is satisfied since only relevant edges are given to the vertices, their values may only increase and no more envy is introduced, and Property (2) is satisfied by the condition of the inner ``while'' (line 9).  

    Hence, we have showed that Properties (1) and (2) are satisfied at the end of Algorithm~\ref{Algo1(2,2)}, and any initially non-envied vertex remains non-envied after Algorithm~\ref{Algo1(2,2)}. The latter automatically means that if any of the Properties (3.1),(3.2) and (3.3) were satisfied before Algorithm~\ref{Algo1(2,2)}, they are preserved after its termination.

    Regarding Property (4), consider first an envied vertex $i$ after the termination of Algorithm~\ref{Algo1(2,2)}. It is trivial to see that Property (4) is satisfied for $i$, because $i$ does not satisfy the condition of line 4, otherwise Algorithm~\ref{Algo1(2,2)} would not have terminated. We now turn our attention to some non-envied vertex $k$ after the termination of Algorithm~\ref{Algo1(2,2)}. Obviously $k$ does not satisfy the condition of line 12 (or the same condition of line 1, if the ``while'' in lines 4-15 has not been executed). This would mean that either i) $v_k(\UNP_k(\X))< v_k(X_k)$, or 
    ii) $v_k(\UNP_k(\X)) = v_k(X_k)$ and $|\UNP_k(\X)| \leq |X_k|$. The former case is not possible due to the definition of $\UNP_k(\X)$ that considers $X_k$ as a possible bundle. So, focusing on the latter case, due to the maximality of $\UNP_k(\X)$, the set allocated to $k$ should include at least one edge related to each neighbor (since the allocation is an orientation and it is not possible to have allocated both edges to other vertices); we add this as an observation next to be used later. Therefore, $U_k(\X)$ is a bundle with no parallel edges, and so  $v_k(X_k)\geq v_k(U_k(\X)) = v_k(U(\X))$.

    \begin{observation}
    \label{obs:AllEdgesNon-enviedAllocated}
        Let $\X$ be the allocation after Algorithm~\ref{Algo1(2,2)}. For any non-envied vertex $k$, $X_k$ contains one edge that is relevant to each of its neighbors.
        This in turn means that there is no unallocated edge in $\X$ where both endpoints are non-envied.
    \end{observation}

    Finally, we argue that Algorithm~\ref{Algo1(2,2)} terminates. Note that the ``while'' at lines 1-3 and 12-14 either strictly increases the social welfare (i.e., the aggregate value of all vertices) or the cardinality of some allocated bundle strictly increases, so they  terminate (in pseudopolynomial time). The ``while'' in lines 9-11 strictly increases the social welfare, so for the same reason it terminates. The ``while'' in lines 4-15, runs at most $n$ times since at each round at least one vertex becomes non-envied and no non-envied vertex becomes envied. 
\end{proof}

\subsection{Step 3 - Final Allocation} 
At this final step, we start by an allocation satisfying Properties (1)-(4) (derived by Step 2) where Property (3) corresponds to one of the Properties (3.1),(3.2),(3.3) related to one of the three different cases, and we properly allocate the unallocated edges, such that the complete allocation is EFX. We remark that the edges allocated at Step 2, remain as they are, and we only allocate once and for all the unallocated edges, to vertices other than the endpoints (so we only extend the bundles allocated in Step 2, in order to derive a full EFX allocation). 

Algorithm~\ref{Algo(FinalAllocation)} finds for each unallocated edge an appropriate non-envied vertex other than its endpoints to allocate that edge without violating EFX; more precisely, none of the endpoints would envy that vertex after the allocation. Note that there only two kind of unallocated edges: unallocated edges between an envied and a non-envied vertex, and unallocated between two envied vertices (due to Observation~\ref{obs:AllEdgesNon-enviedAllocated}).

The difficultly of this final allocation, is to allocate edges with envied endpoints, since we cannot allocate to envied vertices any additional edge, as this would break EFX. This is why we posed the graphical restrictions in Theorems~\ref{theorem:bipartite},~\ref{result2}~and~\ref{result2Girth}, so that there is always a way to allocate those edges without breaking EFX. Therefore, the remaining unallocated edges that are adjacent to some non-envied vertex $i$ are assigned to vertices that are not $i$'s neighbors and therefore $i$ has no value for their allocated bundle and any unallocated edges are not envied (Property (4)). At the same time, the unallocated edges that are adjacent to an envied vertex $i$ are assigned to vertices that $i$ doesn't envy, even if they receive all the unallocated edges adjacent to $i$  additionally to their bundle, under always the restriction of not containing parallel edges (Property (4)). 

\begin{algorithm}
\caption{Complete EFX Allocation}
\label{Algo(FinalAllocation)}
\raggedright\textbf{Input:} An allocation $\X$ satisfying Properties (1)-(4).\\
\textbf{Output:} A complete EFX allocation.
\begin{algorithmic}[1]

\While{$\exists\; e=(i,j)\in U(\X)$} 
    \State Let $k$ be a non-envied vertex such that $X_k \cup \{e\}$ contains no parallel edges, and  
         \Statex \hspace{25pt} for any $\ell\in\{i,j\}$, $v_{\ell}(X_{\ell}) \geq v_{\ell}\left(X_k \cup \{e\}\right)$ 
    \State $X_k \gets X_k\cup \{e\}$
\EndWhile
\end{algorithmic}
\end{algorithm}

In the following lemma we show that Algorithm~\ref{Algo(FinalAllocation)} provides a complete EFX allocation for bipartite multigraphs and for multigraphs under the restrictions in the statements of Theorems~\ref{result2} and~\ref{result2Girth}, which completes the proofs of the theorems for the case of at most two parallel edges per pair of vertices.

\begin{lemma}
For bipartite multigraphs, or for multigraphs where either there are at most $\lfloor \frac{n}{4} \rfloor$ neighbors per vertex, or the shortest cycle with non-parallel edges has length at least 6, Algorithm~\ref{Algo(FinalAllocation)} returns a complete EFX allocation. 
\end{lemma}

\begin{proof}
Let $\X$ be the allocation from Step 2. First note that by Observation~\ref{obs:AllEdgesNon-enviedAllocated} all unallocated edges have at least one envied endpoint. Moreover, by the same observation, for each adjacent vertices where exactly one of them is envied, there may be at most one unallocated edge with those vertices as endpoints. We next show in the following claim that there is always the vertex $k$ of line 2, and moreover, that there are two distinct such vertices for each pair $i,j$ of both envied vertices,  if there are two unallocated edges with those envied endpoints. Then, we argue that all edges will be allocated after running Algorithm~\ref{Algo(FinalAllocation)} by preserving the EFX condition.

\begin{claim}\label{claim: safevertex}
For bipartite multigraphs, or for multigraphs where either the shortest cycle with non-parallel edges has length at least 6, or there are at most $\lfloor \frac{n}{4} \rfloor$ neighbors per vertex, there always exists the vertex $k$ of line 2 in Algorithm~\ref{Algo(FinalAllocation)}. Moreover, if there are two unallocated edges with the same envied endpoints, then there are two {\em distinct} such vertices of line 2.
\end{claim}

\begin{proof}
We first give the following observation to be heavily used in the proof.
\begin{observation}
\label{obs:nonNeighbor-Safe}
    If $i$ is an envied vertex and $k$ is not $i$'s neighbor, then by Property (4), $v_i(X_i) \geq v_i(S) = v_i\left(S \cup X_k\right)$, for all $S\subseteq U(\X)$ s.t. $S \cup X_k$ contains no parallel edges.
\end{observation}
    {\bf Bipartite multigraphs.} We start with the case of bipartite graphs. Note that in this case the only unallocated edges are between an envied and a non-envied vertex (by using also Observation~\ref{obs:AllEdgesNon-enviedAllocated}). Consider any edge $e=(i,j)\in U(\X)$, such that $i$ is envied and $j$ is not envied. We will argue that $p_i(\X)$ is the required $k$ vertex: By Property (3.1) it holds that $p_i(\X)$ is non-envied, and  
    by Property (4), no envy can be created to $i$ by giving any unallocated edges to $p_i(\X)$. Regarding vertex $j$, as an observation note that $j\neq p_i(\X)$, because one edge between $i$ and $p_i(\X)$ has been allocated to $i$ and causing the envy of $p_i(\X)$, and the other has been allocated to $p_i(\X)$ due to Observation~\ref{obs:AllEdgesNon-enviedAllocated}. Since $j$ is adjacent to $i$ and the graph is bipartite, $j$ is not adjacent to $p_i(\X)$, and by Observation~\ref{obs:nonNeighbor-Safe}, no envy can be created to $j$ by giving any unallocated edges to $p_i(\X)$.

    {\bf At most $\lfloor \frac{n}{4} \rfloor$ neighbors.} We proceed with the case that each vertex has at most $\lfloor \frac{n}{4} \rfloor$ neighbors.
    We argue that for any pair of adjacent vertices $i,j$, there are at least two (or one) non-envied vertices that are not adjacent to neither $i$ nor $j$, when $i,j$ are both envied (or exactly one of them is envied). Or in short, if $q$ expresses the number of envied vertices in $\{i,j\}$, we will show that there are at least $q$ non-envied vertices that are not adjacent to either $i$ or $j$. Note that $q$ is also the maximum number of unallocated edges between $i$ and $j$, so this would complete the proof.
    
    The total number of $i$'s and $j$'s neighbors (including $i$ and $j$) is at most $2\lfloor \frac n4 \rfloor \leq \lfloor \frac n2 \rfloor$, so the number of vertices that are not adjacent to either $i$ or $j$ is at least $n-\lfloor \frac n2 \rfloor = \lceil \frac n2 \rceil$. By Property (3.2), among them there are at most $\lfloor \frac n2 \rfloor-q$ envied vertices, so the number of non-envied vertices that are not adjacent to either $i$ or $j$ is at least:
    $\lceil \frac n2 \rceil-\lfloor \frac n2 \rfloor+q \geq q$.

    {\bf Length of shortest cycle with non-parallel edges at least 6.} We now turn our attention to the case that the shortest cycle with non-parallel edges has length at least 6. Consider any edge $e=(i,j)\in U(\X)$, such that $i$ is envied. 

    If $j=p_i(\X)$ then $j$ must be envied: the reason is that $i$ is envied by $j$, so $i$ has received an edge relevant to both $i$ and $j$, and by Observation~\ref{obs:AllEdgesNon-enviedAllocated}, $j$ has received the other edge, contradicting the fact that $e\in U(\X)$. 
    Therefore, if $j=p_i(\X)$, $j$ is envied,  and there exists at most one unallocated edge between $i$ and $j$; the other has been allocated to $i$. In that case, the role of $k$ in line 2 is given to either $p_j(\X)$ if it is non-envied, or to some non-envied neighbor of $p_j(\X)$ (which is guaranteed to exist by Property (3.3) when considering $j$). In both cases, $k$ is not $i$'s neighbor (since the shortest cycle with non-parallel edges has length at least 6). If $k\neq p_j(\X)$, then $k$ is not $j$'s neighbor either, and by Observation~\ref{obs:nonNeighbor-Safe}, $k$ satisfies the conditions of line 2. If $k=p_j(\X)$, then by also using Property (4) for $j$, $k$ satisfies the conditions of line 2.

    If $j\neq p_i(\X)$, by Property (3.3), either $p_i(\X)$ is non-envied or a neighbor of $p_i(\X)$ is  non-envied. In the latter case, that neighbor is defined as $k$ which is not a neighbor of either $i$ or $j$, otherwise a cycle of length less than 6 would exist, which is a contradiction. By Observation~\ref{obs:nonNeighbor-Safe}, $k$ is a vertex satisfying the conditions of line 2. In the former case, we define $p_i(\X)$ as $k$, which is not a neighbor of $j$ for the same reason as above, and it holds $v_i(\UNP_i(\X) \cup X_k)=v_i(\UNP_i(\X))\geq  v_i(S \cup X_k)$, for all $S\in U(\X)$ s.t. $S \cup X_k$ contains no parallel edges. By Property (4) and Observation~\ref{obs:nonNeighbor-Safe}, $k$ is a vertex satisfying the conditions of line 2, if no other edge parallel to $e$ has been given to $k$ before. If $j$ is non-envied, then the statement of the claim holds. We proceed the analysis with the case that $j$ is also envied and there may be two unallocated edges between $i$ and $j$. Let $k_i$ be the $k$ defined above, i.e., the non-envied vertex guaranteed by Property (3.3), by considering the envied vertex $i$. Also let $k_j$ be similarly the non-envied vertex corresponding to vertex $j$; for the same reason as above, $k_j$ satisfies the conditions of line 2. Then, $k_i$ and $k_j$, that have distance at most 2 from $i$ and $j$, respectively, should be different, otherwise there would be a cycle of length at most 5. Then, one edge between $i$ and $j$ may be allocated to $k_i$ and the other one to $k_j$, so the claim follows for that case, as well. 
\end{proof}
In Claim~\ref{claim: safevertex} we showed that for any unallocated edge of $\X$ (the allocation of Step 2) there is always a non-envied vertex satisfying the conditions of line 2, and so Algorithm~\ref{Algo(FinalAllocation)} terminates in a complete allocation. The condition in line 2 guarantees that whenever an unallocated edge is assigned to a vertex $k$, no new envy is caused towards $k$, and since $k$ is non-envied, EFX is preserved. 
\end{proof}

\section{Many Parallel Edges}
\label{sec:generalHighLevel}
We provide a proof roadmap that describes how we tweak our techniques to work on bundles. We first give a definition of forming bundles by using the cut-and-choose protocol; we assume that between any adjacent vertices there are at least two edges, or otherwise w.l.o.g. we may add a dummy edge that do not contribute at all to their value.
\begin{definition} (Bundles)
\label{def:4bundlesHighLevel}
    Let $i,j$ be two adjacent vertices and $E_{ij}$ be the set of their common edges. Based on the cut-and-choose protocol \cite{PlautRough}, we define:
    \begin{itemize}
        \item $B^i_{ji}$ is the bundle that $i$ chooses when $j$ EFX-cuts $E_{ij}$
        
        \item $B^j_{ji} = E_{ij}\setminus B^i_{ji}$, i.e., this is what $j$ gets when $j$ EFX-cuts $E_{ij}$ and $i$ chooses
        
        \item $B^j_{ij}$ and $B^i_{ij}$ are defined accordingly by swapping $i,j$
    \end{itemize}
\end{definition}

We present the properties we satisfy in the general case, where the sets $B_i$ and $UB_i(\X)$ are sets of bundles and will be defined in the related paragraphs. $\UNPB_i(\X)$ is in the same spirit with $\UNP_i(\X)$ of Definition~\ref{def:UNP}:

\begin{tcolorbox}[colback=black!5!white,colframe=black!75!black]
Properties of the allocation $\X$ after the 1st step.
\begin{itemize}
    \item[(1)] A partial EFX orientation, , where for any vertex $i$, $X_i$ is a union of bundles from $B_i$.
    \item[(2)] For any vertex $i$ and $B \in UB_i(\X), v_i(X_i) \geq v_i(B)$.
    \item[(3.1)]  No two envied vertices are adjacent. (For Theorem~\ref{theorem:bipartite})\label{Property: BipartiteGeneral}
    \item[(3.2)] The number of the envied vertices is at most $\lfloor \frac n2\rfloor$. (For Theorem~\ref{result2})\label{Property: BoundedEnviedGeneral}
    \item[(3.3)] For any envied vertex $i$, $p_i(\X)$ or at least one neighbor of $p_i(\X)$ (Definition~\ref{def:p_i(X)}) is non-envied. (For Theorem~\ref{result2Girth})\label{Property: ForGirthGeneral}
\end{itemize}
Additional property of the allocation $\X$  after the 2nd step.
\begin{itemize}
    \item[{(4)}] For any vertex $i$, $v_i(X_i) \geq v_i\left(\UNPB_i(\X)\right)$.
\end{itemize}

\end{tcolorbox}
\vspace{5pt}

{\bf Step 1 for Bipartite multigraphs}

The initial allocation is very similar with the case of at most 2 parallel edges cases. However, we allocate to each vertex in side $A$ their most valued bundle from the partitions formed by the EFX-cuts of the vertices in side $B$.

{\bf Step 1 for bounded number of neighbors}
The initial allocation now allocates bundles to vertices and not just edges. The construction of the bipartite graph $H(G)$ is different regarding to the $B$ vertices, which correspond now to bundles. There are two important cases while considering the bundles of side $B$, and those are whether for two vertices $i$ and $j$ in $G$ there exists a common EFX-cut for their common relevant edges or not. If yes, we add on side $B$ of $H(G)$ the two bundles of this common EFX-cut, otherwise we construct a three partition of their common edges (and we add those bundles on side $B$ of $H(G)$) such that $i$ and $j$ have a distinct most valued bundle among them; in Lemma~\ref{lemma:notsamepartition} we show that such a partition always exists. In similar fashion as in Section~\ref{sec:2parallel}, we connect each vertex in side $A$ with their most and second most valued bundle of side $B$ with weights $1$ and $0$, respectively. The same arguments as in Section~\ref{sec:2parallel} stand in order to show that the allocation corresponding to a maximum $A$-perfect matching satisfies Properties (1),(2),(3.2). 

We remark that in the case of at most two parallel edges each pair of neighbors had a common EFX-cut. In this more general case of many parallel edges, we also require the partition of some parallel edges into three bundles instead of two. This distinction is responsible for requiring $\lceil \frac{n}{4} \rceil -1$ neighbors (instead of $\lfloor \frac{n}{4} \rfloor$ as in the case of two parallel edges) because we now need more ``safe'' vertices to ``park'' those extra bundles.

{\bf Step 1 for multigraphs where the shortest cycle using non-parallel edges is of length at least 6}

The initial allocation now allocates bundles to vertices and not just edges. We utilize a greedy approach close to Section~\ref{sec:2parallel} unified with a similar algorithm to Algorithm~\ref{Algo1(InitialThmGirth)}. The problem of neighbors not having the same EFX-cut appears here as well, but we deal with it in a different way (as we cannot guarantee the existence of that many ``safe'' vertices): we stick to only partitioning the parallel edges with an EFX-cut, but we may consider both EFX-cuts of both neighbors. Properties (1) and (2) are guaranteed by offering all those possible bundles by prioritizing on the endpoints. In order to guarantee property (3.3) we follow the idea of Algorithm~\ref{Algo1(InitialThmGirth)} and for every new non-envied vertex we create we ``lock'' their relevant bundles to the ones that correspond to the case that they EFX-cut, preserving that way that they will remain non-envied.

\begin{remark}
{\em We remark that using the approach of case 2, i.e. the three partitions between some pair of vertices, in case 3 would mean that we may require three non-envied vertices to ``park'' those bundles in case they remain unallocated before the final allocation. However, Property (3.3) and the restriction of this case (i.e., all cycles of non-parallel edges have length at least 6), guarantee the existence of only two such vertices, and therefore we have to restrict ourselves to two partitions only. 

On the other hand, using the approach of case 3 in case 2 may be also problematic. Suppose that in the bipartite graph $H(G)$ we use all bundles formed by EFX-cuts (i.e., for vertices $i,j$ we consider $B_{ij}^i, B_{ij}^j, B_{ji}^i, B_{ji}^j$ on the B side). The question is if we would considering assigning e.g., $B_{ji}^i$ to $j$ in the maximum $A$-perfect  matching (i.e., if edges $(j,B_{ji}^i)$ would be added in $H(G)$). If we allow such an edge, $j$ could be assigned $B_{ji}^i$ and $i$ could possibly not be EFX-satisfied against $j$. If we do not allow such an edge, it may be that $i$ is assigned $B_{ji}^i$, and $j$ would be EFX-satisfied (as required for Property (1)). Nevertheless, $j$ may envy $i$, which is an envy not expressed in $H(G)$, and therefore, the maximum $A$-perfect matching in $H(G)$ cannot guarantee limited envy (as in the case of at most two parallel edges). 

So, overall, it seems necessary to use different approaches for the initial allocations for those two cases.}
\end{remark}

{\bf Step 2 - Satisfying Property (4).}
This step is essentially the same with Algorithm~\ref{Algo1(2,2)}. The algorithm works in the exact same way, but now it considers non-parallel bundles instead of single edges. 
There is only a little caution with respect to the priority of agents while offering them bundles, similarly to Step 1 regarding Theorem~\ref{result2Girth}.\\

{\bf Step 3 - Final Allocation.}
Algorithm~\ref{Algo(FinalAllocation)} now considers bundles instead of single edges.
The arguments in Claim~\ref{claim: safevertex} are exactly the same with the general case. In the case where the shortest cycle with non-parallel edges has length at least 6 or in bipartite graphs, the final allocation proceeds by just considering at most two bundles between two vertices, instead of at most two edges, but the arguments remain the same. 

In the case of bounded number of neighbors, in the general case there may be one more unallocated bundle between two vertices compared to the case of two parallel edges, i.e., three bundles when both endpoints are envied, two bundles when only one endpoint is envied, and one bundle when both endpoints are non-envied. This is dealt by restricting the number of possible neighbors by at most one more, or more specifically, in the case of two parallel edges, the requirement was to have at most $\lfloor \frac{n}{4} \rfloor$ neighbors, whereas, in the general case we assume at most $\lceil \frac{n}{4} \rceil -1$ neighbors. \\

\subsection{Step 1 - Initial Allocations}

In this step we construct the initial allocation with specific properties that are crucial towards the final EFX construction. We construct an initial allocation that satisfies Properties (1),(2) and respectively one of the Properties (3.1),(3.2) and (3.3) for each of our main theorems.

\vspace{5pt}

\begin{tcolorbox}[colback=black!5!white,colframe=black!75!black]
Properties of the allocation $\X$ after the 1st step.
\begin{itemize}
    \item[(1)] A partial EFX orientation, where for any vertex $i$, $X_i$ is a union of bundles from $B_i$. 
    \item[(2)] For any vertex $i$ and $B \in UB_i(\X), v_i(X_i) \geq v_i(B)$.
    \item[(3.1)] No two envied vertices are adjacent. (For Theorem~\ref{theorem:bipartite})
    \item[(3.2)] The number of the envied vertices is at most $\lfloor \frac n2\rfloor$. (For Theorem~\ref{result2})
    \item[(3.3)] For any envied vertex $i$, $p_i(\X)$ or at least one neighbor of $p_i(\X)$ is non-envied. (For Theorem~\ref{result2Girth}) 
\end{itemize}
\end{tcolorbox}

\vspace{5pt}

    {\bf The sets $B_i$ and $UB_i(\X)$.} In each of the three cases we define for each agent $i$ a set of  bundles $B_i$ that are adjacent to $i$, and could be allocated as a whole to them. We require that 
    \begin{enumerate}
        \item the set $\cup_{S \in B_i}S$ is the set of all edges relevant to $i$, and 
        \item for any adjacent vertices $i,j$, we specify a single partition of their common relevant edges, and only the bundles of this partition belong to $B_i$ and $B_j$ w.r.t. this pair of vertices.
    \end{enumerate}
    Then, $UB_i(\X)$ is the subset of unallocated bundles of $B_i$ under some allocation $\X$.

\subsubsection{Initial allocation for Bipartite Multigraphs}

We first define the sets $B_i$ for the case of bipartite multigraphs.

\begin{definition}\label{bundlesdef1} ($B_i$ set for Theorem~\ref{theorem:bipartite}) For any vertex $i$, we define $B_i$ as follows:
\begin{itemize}
    \item If $i \in A$, then $B_i = \{B^i_{ji}, B^j_{ji}\forall j \text{ adjacent to } $i$\}$
    \item If $i \in B$, then $B_i = \{B^i_{ij},B^j_{ij}, \forall j \text{ adjacent to } $i$\}$
\end{itemize}

Note that all edges between two endpoints are partitioned in a single way, i.e., the vertex in side $B$ EFX-cuts. 
\end{definition}

Then, the initial allocation for the case of bipartite multigraphs is constructed as follows: Starting from the empty allocation, each vertex $i\in A$ is allocated their most valued bundle from $B_i$, i.e., their most valued bundle when their neighbors EFX-cut their common edges. Then, each vertex in side $B$ is allocated their most valued bundle from the remaining unallocated bundles. This procedure is formally described in Algorithm~\ref{Algo:InitBipar2}.

\begin{algorithm}
\caption{Initial Allocation for Bipartite Graphs}
\label{Algo:InitBipar2}
\raggedright\textbf{Input:} A bipartite multigraph $G = (V = A \cup B,E)$. \\
\textbf{Output:} $\X$ satisfying Properties (1),(2) and (3.1). %
\begin{algorithmic}[1]
\For{every vertex $i$}
    \State $X_i \gets \emptyset$
\EndFor
\For{every $i \in A$}
    \State $X_i \gets S_i \in \arg \max_{S \in UB_i(\X)}\{v_i(S)\}$
\EndFor

\For{every $j \in B$}
    \State $X_j \gets S_j\in \arg \max_{S \in UB_j(\X)}\{v_j(S)\}$
\EndFor
\end{algorithmic}
\end{algorithm}

In the next lemma we show that the allocation derived by Algorithm~\ref{Algo:InitBipar2} indeed satisfies Properties (1),(2) and (3.1).

\begin{lemma}
    Algorithm~\ref{Algo:InitBipar2} outputs an allocation $\X$ that satisfies Properties (1), (2) and (3.1).
\end{lemma}

\begin{proof}
     Property (1) is satisfied for all vertices in side $A$, since they get their most valued adjacent bundle. Vertices in side $B$ may envy some vertices from side $A$, however, the allocation is EFX since they EFX-cut their relevant edges with any of their neighbors: More formally, consider a vertex $j\in B$ that envies vertex $i\in A$. Then, $i$ should have received $B_{ji}^i$, and $B_{ji}^j$ remains available until the time that $j$ is allocated a bundle. Vertex $j$ will then receive a bundle that she values by at least as much as $B_{ji}^j$. Since $j$ would be EFX-satisfied against $i$ if she received $B_{ji}^j$, she will be EFX-satisfied against $j$ in $\X$. 
     
     Property (2) is also trivially satisfied since when Algorithm~\ref{Algo:InitBipar2} updates the bundle allocated to any vertex $i$, they receive their most valued bundle from $UB_i(\X)$, so they prefer $X_i$ to any unallocated bundle in $UB_i(\X)$. 
     
     For Property (3.1) note that all vertices in side $A$ are assigned their most valued bundle before assigning any bundle to side $B$. Since the graph is bipartite, only vertices from side $B$ may envy vertices from side $A$. Therefore, all envied vertices are on side $A$, and in turn, no two envied vertices are adjacent.
\end{proof}

\begin{observation}
\label{obs:Case1NoUnallocated}
    After Algorithm~\ref{Algo:InitBipar2}, between any two adjacent vertices there are at most two unallocated bundles.
\end{observation}

\subsubsection{Initial allocation for at most $\lceil n/4 \rceil$ - 1 neighbors per vertex} 
One useful property in the case of multigraphs with {\em at most two} parallel edges is that every pair of vertices, in some way ``agree'' on how to EFX-cut their common edges (separate the two parallel edges is a EFX-cut for both vertices). However, when there are three or more parallel edges between two vertices, it may be the case that no EFX-cut for one vertex is an EFX-cut for the other. This would cause an issue to our matching allocation if we follow the approach of Section~\ref{sec:2parallel}. 

Specifically, in the general case, for any adjacent vertices $i,j$, the bundle $B_{ji}^i$, that $i$ chooses when $j$ EFX-cuts, may have high value for $i$, so if it is allocated to $j$, it may break EFX. On the other hand, if that bundle is not offered to $j$ in the matching allocation, vertices will be allocated their most or second most valued bundle from a {\em restricted} set of bundles. In that case vertex $j$ cannot receive $B_{ji}^i$ in any matching (since there will be no edge $(j,B_{ji}^i)$ in $H(G)$). It may though be the case that $B_{ji}^i$ is $j$'s most valuable bundle, and if $i$ is allocated that bundle, then $j$ would envy $i$. However, $j$ could have received her most valued bundle in the maximum $A$-perfect matching, but it is not true anymore that she does not envy others vertices, as we argued in Section~\ref{sec:2parallel} in order to show Property (3.2). 

To accommodate this issue, we construct a 3-partition of the edges between vertices without common EFX-cuts (Lemma~\ref{lemma:notsamepartition}), for which each endpoint has different most valued set among the three. We next distinguish between pair of vertices that admit a common EFX-cut or not.

\begin{definition}\label{def:partition1Case2} (PVCP, $B^1_{i(j)}, B^2_{i(j)}$)
    Let $i,j$ be two adjacent vertices and $E_{ij}$ be the set of their common edges. If there exists a partition $(P_1,P_2)$ of $E_{ij}$ that either corresponds to an envy-free allocation to $i$ and $j$ or $i$ and $j$ are EFX-satisfied against each other no matter how the sets are allocated to them, (i.e, $(P_1,P_2)$ is an EFX-cut for both $i$ and $j$), then we denote by $B^1_{i(j)}, B^2_{i(j)}$, 
    $i$'s most and second most valued bundles, respectively, between $P_1$ and $P_2$. Similarly, we define $B^1_{j(i)}, B^2_{j(i)}$, to be $j$'s most and second most valued bundles, respectively, between $P_1$ and $P_2$. We further define PVCP to be the set of Pair of Vertices whose common edges admit such a Common Partition as described above.
\end{definition}

Next we handle the rest of the edges, and we show that for each pair of vertices $(i,j) \notin \text{PVCP}$ there exists a partition $(P_1, P_2, P_3)$ of their common edges $E_{ij}$ into three bundles, such that $i,j$ have different top preference among them. 

\begin{lemma}\label{lemma:notsamepartition}
For any pair of vertices $(i,j) \notin \text{PVCP}$ there exists a partition $(P^1_{ij}, P^2_{ij}, P^3_{ij})$ of their common edges $E_{ij}$ into three bundles, such that $i,j$ have different most valued set among them.
\end{lemma}

\begin{proof}
    Since $(i,j) \notin \text{PVCP}$, for every partition of $E_{ij}$ into two bundles, $i,j$ have the same top preference. Suppose that $(P_1,P_2)$ is an EFX-cut for $i$, and let $P_2$ be the top preference between $P_1$ and $P_2$ for both $i$ and $j$. Moreover, $i$ is EFX-satisfied with $P_1$ as well, whereas $j$ is not (because $(i,j) \notin \text{PVCP}$). This means that there exists a good $g\in P_2$, such that $v_j(P_2\setminus\{g\}) >v_j(P_1)$, and $v_i(P_2\setminus\{g\}) \leq v_i(P_1)$. Additionally, $v_i(\{g\}) \leq v_i(P_1)$, because if $P_2$ contains more than one goods, $i$ should prefer $P_1$ to $P_2$ after the removal of any other good but $g$ from $P_2$. It cannot be that $|P_2|=1$, because then $j$ would be EFX-satisfied with $P_1$, as well. So, in the partition $(P_1, P_2\setminus\{g\}, \{g\})$, $i$'s top preference is $P_1$ and $j$'s top preference is not $P_1$. Therefore, for $(P^1_{ij}, P^2_{ij}, P^3_{ij})=(P_1, P_2\setminus\{g\}, \{g\})$, the lemma follows. 
\end{proof}

The above partitions of the edges are next used in order to define $B_i(\X)$, and therefore $UB_i(\X)$, for any vertex $i$.

\begin{definition}
    We define the sets $B_i$ for all vertices $i$ as follows: for any pair of adjacent vertices $i$ and $j$,
    \begin{itemize}
        \item if $(i,j) \in \text{PVCP}$, $B^1_{i(j)}, B^2_{i(j)} \in B_r$, for all $r\in\{i,j\}$ (see Definition~\ref{def:partition1Case2}),
        
        \item if $(i,j) \notin \text{PVCP}$, $P^1_{ij}, P^2_{ij}, P^3_{ij} \in B_r$, for all $r\in\{i,j\}$ (see Lemma~\ref{lemma:notsamepartition}).
    \end{itemize} %
Note that all edges between two endpoints are partitioned in a single way, i.e., the vertex in side $B$ EFX-cuts. 
\end{definition}

We show that a similar matching from Section~\ref{sec:2parallel} gives us the three properties to proceed to Step 2.

{\bf Satisfying Properties (1), (2) and (3.2).}
The initial allocation is constructed similarly to the one in Section~\ref{sec:2parallel}; we construct a similar bipartite graph, however the main difference is that on the $B$-side we consider bundles instead of single edges. The tricky part is to recognize the right bundles to consider. This the reason why we distinguish between pairs of neighbors with common EFX-cut (the PVCP) and the rest that do not have any common EFX-cut. In the former case we consider the two bundles of the endpoints' common partition (see Definition~\ref{def:partition1Case2}), and in the latter case  we consider the 3-partition as defined in Lemma~\ref{lemma:notsamepartition}. Next we formally define the bipartite graph.
\begin{definition}\label{HmultigraphNeighbors}(Bipartite graph $H(G)$).
     Given a multigraph $G=(V,E)$, we define a simple weighted bipartite graph $H(G) = (A \cup B,E_H)$ as follows. The set $A$ represents the vertices of $G$, so $A=V$. The set $B$ consists of all the bundles formed by partitioning the edges into two or three bundles as described above , i.e., $B=\bigcup_iB_i$. 
     We create the edge set $E_H$ as follows: we connect each vertex $i$ of $A$ (i.e., each vertex of $G$) with the two vertices of $B$ that are associated with $i$'s most valued and second most valued bundle from the $B$-side. We set the weights of those two edges we add in $H(G)$ to $1$ and $0$, respectively. 
     \end{definition}

 We prove in Lemma~\ref{2,2matchInf} that the allocation, let it be $\XX$, that corresponds to a maximum $A$-perfect matching
  in $H(G)$, satisfies Properties (1), (2), and (3.2). Note that it can be easily shown that there is always an $A$-perfect matching
  in $H(G)$ by following the arguments of Observation~\ref{obs:existence_PerfectMatching}

\begin{lemma}\label{2,2matchInf}
    $\XX$ satisfies Properties (1), (2) and (3.2).
\end{lemma}

\begin{proof}
    Regarding Property (1), suppose that $i$ envies $j$ in $\XX$, then $(i,j)\in \text{PVCP}$ and $j$ should have received $B^1_{i(j)}$. So, it holds that $i$ receives a bundle in $\XX$ at least as valued as $B^2_{i(j)}$. Since $i$ would be EFX-satisfied against $j$ if it received $B^2_{i(j)}$, $i$ is EFX-satisfied against $j$ in $\XX$.

    Regarding Property (2), note that for every vertex $i$, all bundles of $B_i$ are included on the $B$ side of $H(G)$ (i.e., $B_i \subseteq B$). Then, $i$ receives its most or second most valued bundle from $B_i$. If $i$ receives its most valued bundle, then Property (2) is trivially satisfied. If $i$ receives its second most valued bundle, this means that its most valued bundle is allocated to another vertex due to the maximality of $\XX$. Therefore, Property (2) follows. 
    
    Regarding Property (3.2), with similar arguments as in the proof of Lemma~\ref{2,2match}, each agent gets their most or second most valued bundle from side $B$. Again each vertex representing an agent has degree exactly 2, thus they participate in exactly one alternating cycle or path. Since the matching is a maximum weighted matching, at least half of the vertices, $\lceil \frac n2\rceil$, are matched with their top preference from side $B$. So, at most half of them envy another vertex and at most one.  Therefore, at most $\lfloor \frac n2\rfloor$ vertices are envied.   
\end{proof}

\begin{observation}
\label{obs:Case2NoUnallocated}
    In $\XX$, between any two adjacent vertices there are at most three unallocated bundles.
\end{observation}

\subsubsection{Initial allocation for length of the shortest cycle with non-parallel edges at least 6}

In contrast to the previous case, we need to make sure that between any two envied vertices there are at most two unallocated bundles, and therefore, we cannot use the three partition of Lemma~\ref{lemma:notsamepartition}. Instead we stick to the EFX-cuts so that only two bundles are formed. However, it may be the case that vertices do not share a common EFX-cut, and it is quite tricky to decide which to use without violating Properties (1) and (2). Therefore, at the beginning we cannot define the set $B_i$ for every vertex $i$, as opposed to the previous cases. Instead, for any pair of vertices, we initially consider both partitions where each of the endpoints EFX-cuts their relevant edges. We finalize the sets $B_i$ along with finalizing the initial allocation for this case. We note that for the finalized $B_i$ sets, for any pair of adjacent vertices $i$ and $j$, only one of the two partitions, $(B_{ij}^i,B_{ij}^j)$ or $(B_{ji}^i,B_{ji}^j)$, belongs to both $B_i$ and $B_j$, i.e., for both $B_i$ and $B_j$ we will fix which endpoint EFX-cuts their common edges.

Hence, we define some auxiliary $B_i$'s that will be finalized later. 

\begin{definition}
\label{def:auxB_iCase3}
    For each vertex $i$, we define the auxiliary set $B_i$ as:
        $$B_i^{\text{aux}} = \{B^i_{ij},B^j_{ij},B^i_{ji},B^j_{ji} \mid  \forall j \text{ adjacent to } i\}.$$
       We  further defined the auxiliary set of unallocated bundles $UB_i^{\text{aux}}(\X)$ as follows: for each adjacent vertices $i,j$, 
       \begin{itemize}
           \item if the bundles $B_i^{\text{aux}}\cap B_j^{\text{aux}}$ are all unallocated in $\X$, then $B_i^{\text{aux}}\cap B_j^{\text{aux}} \subseteq UB_i^{\text{aux}}(\X)$,
           \item if exactly one of the bundles in $B_i^{\text{aux}}\cap B_j^{\text{aux}}$ is allocated in $\X$, i.e.,  only $B_{rr'}^r$ or $B_{rr'}^{r'}$, for $r\neq r'$ and $r,r'\{i,j\}$, is allocated in $\X$, then  $B_{rr'}^{r'} \in UB_i^{\text{aux}}(\X)$ or $B_{rr'}^{r}\in UB_i^{\text{aux}}(\X)$, respectively, and $B_{r'r}^{r'}, B_{r'r}^{r} \notin UB_i^{\text{aux}}(\X)$.
           \item in any other case, none of the bundles in $B_i^{\text{aux}}\cap B_j^{\text{aux}}$ are included in $UB_i^{\text{aux}}(\X)$.
       \end{itemize}
         
\end{definition}

{\bf Satisfying Properties (1), (2) and (3.3).}

Similarly to the case with only 2 parallel edges, we first consider some initial allocation that satisfies Properties (1) and (2), by ignoring how many envied vertices there may be. Then, we transform it by creating non-envied vertices to satisfy Property (3.3). In order to preserve that those vertices remain non-envied, in the sets $B_i$ we consider only the bundles formed when those vertices EFX-cut the relevant edges with each of their neighbors. Therefore, the sets $B_i$, that are finalized at the end of the algorithm, satisfy all required properties. 

Algorithm~\ref{Algo:InitSimpleGreedyINF} starts by constructing an EFX orientation that satisfies Properties (1) and (2), where for Property (2) we consider the sets $B_i^{\text{aux}}$ that are changing during the algorithm. Note that if for some pair of vertices $i,j$, vertex $j$ is free to get $B_{ji}^i$ at any point, $i$ may not be EFX-satisfied against $j$ and Property (1) would be violated. On the other hand, if $j$ did not have the chance to get $B_{ji}^i$, Property (2) may be violated for $j$. To address this issue, we just ``offer'' $B_{ji}^i$ first to $i$ and then to $j$, and this prioritization of the endpoints for offering them specific bundles is sufficient in order to satisfy both Properties (1) and (2).   We follow the same tactic any time we convert some envied vertex into non-envied to satisfy Property (3.3). As we mentioned above we make sure that these vertices will remain non-envied throughout the execution of the algorithm by changing (or more accurately shrinking) the sets $B_i^{\text{aux}}$. At the end of the algorithm those sets become the finalized sets $B^i$.

\begin{algorithm}
\caption{Initial Allocation for Theorem~\ref{result2Girth}}
\label{Algo:InitSimpleGreedyINF}
\raggedright\textbf{Input:} A multigraph $G = (V,E)$ and $B_i^{\text{aux}}$ for each $i\in V$. \\
\textbf{Output:} $\X$ satisfying Properties (1), (2) and (3.3) and $B_i$ for each $i\in V$. %
\begin{algorithmic}[1]
\For{every vertex $i$}
    \State $X_i \gets \emptyset$
\EndFor

\While{$\exists i,j$ and $S\in \{B^i_{ji},B^j_{ji}\}$ where $S \in UB^{\text{aux}}_i(\X)$, such that $v_i(S) > v_i(X_i)$ or $v_j(S) > v_j(X_j)$}
    \If {$B^i_{ji}\in UB^{\text{aux}}_i(\X)$ and  $v_i(B^i_{ji}) > v_i(X_i)$}
     $X_i \gets B^i_{ji}$
    \Else \; $X_j \gets S$
    \EndIf
\EndWhile 

\While{$\exists$ envied vertex $a$ for which Property (3.3) is not satisfied} 
    \State Let $b=p_a(\X)$ and $c=p_b(\X)$  
    \State $X_b \gets \emptyset$
    \State $B_b^{\text{aux}} \gets \{B^b_{br}, B^r_{br}\mid  \forall r \text{ adjacent to }b\}$
    \For{ Every vertex $r$ adjacent to $b$}
        \State $B_r^{\text{aux}} \gets B_r^{\text{aux}} \setminus \{B^b_{rb}, B^r_{rb}\}$
    \EndFor
    \While{$\exists i,j$ and $S\in \{B^i_{ji},B^j_{ji}\}$ where $S \in UB^{\text{aux}}_i(\X)$, such that $v_i(S) > v_i(X_i)$ or $v_j(S) > v_j(X_j)$}
    \If {$B^i_{ji}\in UB^{\text{aux}}_i(\X)$ and  $v_i(B^i_{ji}) > v_i(X_i)$}
     $X_i \gets B^i_{ji}$
    \Else \; $X_j \gets S$
    \EndIf
\EndWhile
\EndWhile
\For{ every vertex $i$}
    \For { every vertex $j$ adjacent to $i$}
        \If{$B^r_{ij} \in X_i\cup X_j$, for some $r \in \{i,j\}$, or it holds that $\{B^i_{ij}, B^j_{ij}, B^i_{ji}, B^j_{ji}\}\subseteq UB^{\text{aux}}_i(\X)$}
            \State $B_i^{\text{aux}} \gets B_i^{\text{aux}}\setminus \{B^i_{ji},B^j_{ji}\}$
            \State $B_j^{\text{aux}} \gets B_j^{\text{aux}}\setminus \{B^i_{ji},B^j_{ji}\}$
        \EndIf
    \EndFor
\EndFor
\For{ every vertex $i$}
    \State $B_i \gets B_i^{\text{aux}}$
\EndFor

\end{algorithmic}
\end{algorithm}

Next we show that Algorithm~\ref{Algo:InitSimpleGreedyINF} terminates in an allocation satisfying Properties (1), (2) and (3.3). 
\begin{lemma}
    The allocation returned by Algorithm~\ref{Algo:InitSimpleGreedyINF} satisfies Properties (1), (2) and (3.3).
\end{lemma}

\begin{proof}

    We first give a useful observation:
    \begin{observation}
    \label{obs:Non-enviedWithForbidden}
        For each pair of vertices $i,j$, if at any step of Algorithm~\ref{Algo:InitSimpleGreedyINF}, vertex $j$ receives $B^i_{ji}$, then $j$ is not envied. 
    \end{observation}
    \begin{proof}
        The reason is that $j$ may receive $B^i_{ji}$ only in lines 6 and 18, in which cases $v_i(X_i)\geq v_i(B^i_{ji})$. Further note that the only case that a value of an agent may decrease is in line 11. We argue that $i$ cannot be the vertex $b$ (line 10), if $X_j = B^i_{ji}$. The reason is that $i$ may be vertex $b$, only when all of its neighbors are envied. However, $i$ has a neighbor, namely $j$, that is not envied. 
    \end{proof}
    
    We start by showing that after the termination of the first ``while'' (lines 4-8) Properties (1) and (2) are satisfied. Consider any vertex $j$ envied by vertex $i$, after the termination of the first ``while''. By Observation~\ref{obs:Non-enviedWithForbidden}, vertex $j$ is not assigned $B^i_{ji}$, therefore it is assigned 
    $B^j_{ij}$. Then $i$ gets a bundle of value at least of $B^i_{ij}$: either $B^i_{ij}$ is assigned to $i$, or $B^i_{ij}$ is available and since the ``while'' condition (line 4) is not satisfied, $i$ receives a bundle at least as good. So, $i$ is EFX-satisfied against $j$. Property (2) is trivially satisfied at the end of the first ``while'' since the condition of line 4 is not true. Note that this ``while'' will terminate since the social welfare strictly increases at each iteration.

    Similarly, Properties (1) and (2) hold at the end of the second ``while'' (lines 9-21), since its last subroutine (lines 16-20) is the same with the first ``while'' (lines 4-8). Regarding Property (3.3), this is trivially satisfied at the end of the second ``while'' since the condition in line 9 is not satisfied.

    Regarding the termination of the algorithm, it is sufficient to show that whenever a vertex $b$ becomes non-envied, it remains that way. This would in turn mean that no vertex $a$ can satisfy the ``while'' condition (line 9) more than once (the termination of the two ``while'' in lines 4-8 and 16-20 is discussed above). When a vertex $b$ is considered, we shrink the set $B_b^{\text{aux}}$ to contain only the two bundles formed when $b$ EFX-cuts the common edges with each of her neighbors, and we keep the same bundles in the $B_r^{\text{aux}}$ set of each neighbor $r$ of $b$. By Observation~\ref{obs:Non-enviedWithForbidden}, $b$ cannot be envied as long as $B_b^{\text{aux}}$ does not change any further. This cannot happen, because the $B_i^{\text{aux}}$ sets may only alter with respect to bundles where both endpoints were envied before the ``while'' (this was the case with $b$ and all its neighbors $r$), and clearly $b$ is not envied anymore. Therefore, since vertex $b$ will remain non-envied until the end of the ``while'', vertex $a$ cannot be considered again. 

    Finally, we need to show that at each step of the algorithm it holds $X_i\in B_i^{\text{aux}}$ for all $i$. $B_i^{\text{aux}}$ may change inside the second ``while'' (lines 12-15), and in lines 22-29. In the latter case, $X_i\in B_i^{\text{aux}}$ holds for all $i$ due to the ``if'' statement in line 24.\footnote{Note that it is not possible that two adjacent vertices $i,j$, receive both bundles from different EFX-cuts, e.g., $X_i=B_{ij}^i$ and $X_j=B_{ji}^j$, because if for instance $i$ first receives $B_{ij}^i$, then it holds that $B_{ji}^j\notin UB^{\text{aux}}_j(\X)$ based on Definition~\ref{def:auxB_iCase3}.} In the former case, we need to show that for any neighbor $r$ of $b$, $X_r$ cannot be $B_{rb}^b$ or $B_{rb}^r$, so $X_r$ belongs to the updated $B_r^{\text{aux}}$. If $X_r=B_{rb}^b$, by Observation~\ref{obs:Non-enviedWithForbidden}, $r$ was non-envied before the ``while'', and the same holds if $X_r=B_{rb}^r$, because $b$ values $X_b$ at least as much as $B_{rb}^b$, which in turn values at least as much as $B_{rb}^r$. In any case $r$ would be non-envied, but this cannot happen because then $a$ would not satisfy the ``while'' condition (line 9). Therefore, $X_r$ cannot be neither $B_{rb}^b$ nor $B_{rb}^r$.

    The last for-loops (lines 22-29) guarantee the properties for the $B_i$ sets. Overall, the algorithm terminates by satisfying Properties (1), (2) and (3.3). 
\end{proof}

\begin{observation}
\label{obs:Case3NoUnallocated}
    After Algorithm~\ref{Algo:InitSimpleGreedyINF}, between any two adjacent vertices there are at most two unallocated bundles.
\end{observation}

The rest of the steps are the same for all three cases. 

\subsection{Step 2 - Satisfying Property (4)} 
This step starts with the initial allocation constructed in Step 1. In Step 2, the initial allocation changes according to Algorithm~\ref{Algo1(2,inf)} in order to additionally satisfy Property (4), which needs the following definition:

\begin{definition}\label{def:UNPB} 
    For any envied vertex $i$, let $T=\{X_{p_i(\X)}\} \cap B_i$ be the edges that $p_i(\X)$ receives in $\X$ that are adjacent to $i$, if any (if $\X$ is an initial allocation from Step 1, then $T$ is either a bundle from $B_i$ that $p_i(\X)$ gets, or the empty set, if $p_i(\X)$ gets edges irrelevant to $i$). Let $\hat{U}B_i(\X)=UB_i(\X)\cup \{T\}$. We define the most valued set of potentially unallocated non-parallel bundles as
    \begin{eqnarray*}
        &&\UNPB_i(\X) \in 
        \quad \argmax_{S \subseteq \hat{U}B_i(\X)} \{v_i(S)| \text{ any $B_1,B_2 \in S$ are not parallel}\}\,\text{\footnotemark}.
    \end{eqnarray*}
    \footnotetext{Recall that $v_i(S)=v_i(\cup_{B\in S} B)$}
    where $\UNPB_i(\X)$ is chosen to have the {\em maximum possible cardinality} with respect to the number of bundles from $\hat{U}B_i(\X)$. This is the best set of non-parallel bundles that $i$ could get if $X_i$ would be given to $p_i(\X)$ (who envies $i$), and $X_{p_i(\X)}$ becomes available.
\end{definition}

\begin{definition}\label{def:UNPBnonenvied} 
    For any non-envied vertex $i$, let $\hat{U}B_i(\X)=UB_i(\X)\cup \{X_i\}$. We define the most valued set of potentially unallocated non-parallel bundles as
    \begin{eqnarray*}
        &&\UNPB_i(\X) \in 
        \quad \argmax_{S \subseteq \hat{U}B_i(\X)} \{v_i(S)| \text{ any $B_1,B_2 \in S$ are not parallel}\}\,.
    \end{eqnarray*}
    where $\UNPB_i(\X)$ is chosen to have the {\em maximum possible cardinality} (considering the number of bundles from $\hat{U}B_i(\X)$). This is the best set of non-parallel bundles that $i$ could get without changing the allocation of the other vertices.
\end{definition}

\begin{tcolorbox}[colback=black!5!white,colframe=black!75!black]
Additional property of the allocation $\X$  after the 2nd step.
\begin{itemize}
    \item[{(4)}] For any vertex $i$, $v_i(X_i) \geq v_i\left(\UNPB_i(\X)\right)$.
\end{itemize}
\end{tcolorbox}

The necessity of Property (4) is so that we can allocate  the remaining unallocated bundles at Step 3 in non-adjacent vertices. Note that for any envied vertex $i$ we cannot allocate any of its adjacent unallocated bundles  to $i$ additionally to what it has. So, the high level idea of Step 2 is that either $UB_i(\X)$, or more accurately $\UNPB_i(\X)$, is valued for $i$, so we allocate it to $i$, and $i$ releases the bundle he had which caused the envy, or $i$ doesn't value it that much so we can give $\UNPB_i(\X)$, or subsets of it, to other agents. 
 In the same spirit, the non-envied vertices receive the best set of their adjacent available non-parallel bundles, so that allocating the rest adjacent bundles to other vertices would preserve EFX. The idea is similar to Algorithm 2 of \cite{EFXsimplegraphs}, but adjusted in the multigraph setting.

\begin{algorithm}
\caption{Reducing Envy Algorithm}
\label{Algo1(2,inf)}
\raggedright\textbf{Input:} An allocation $\X$ satisfying Properties (1) and (2). \\
\textbf{Output:} An allocation satisfying Properties (1), (2) and (4). The allocation also preserves any of the Properties (3.1),(3.2),(3.3) if they were initially satisfied.
\begin{algorithmic}[1] %
\While{ $\exists$ non-envied vertex $k$ s.t. $v_k(\UNPB_k(\X))> v_k(X_k)$ or 
\State ($v_k(\UNPB_k(\X)) = v_k(X_k)$ and $|\UNPB_k(\X)| > |X_k|)$}
    \State $X_k \gets \UNPB_k(\X)$ 
\EndWhile

\While{$\exists$ envied vertex $i$ s.t. $v_i(\UNPB_i(\X))> v_i(X_i)$}
    \If{$\{X_{p_i(\X)}\} \cap B_i \in \UNPB_i(\X)$} 
        \State $X_{p_i(\X)} \gets \emptyset$
    \EndIf
    \State $X_i \gets \UNPB_i(\X)$ 
    \While{$\exists$ vertex $j$ and $S \in UB_i(\X)$ s.t. $v_j(S) > v_j(X_j)$} 
        \State $X_j \gets S$ 
    \EndWhile
    \While{ $\exists$ non-envied vertex $k$ s.t. $v_k(\UNPB_k(\X))> v_k(X_k)$ or 
\State ($v_k(\UNPB_k(\X)) = v_k(X_k)$ and $|\UNPB_k(\X)| > |X_k|)$}
        \State $X_k \gets \UNPB_k(\X)$ 
    \EndWhile
\EndWhile
\end{algorithmic}
\end{algorithm}

\begin{lemma}\label{terminationAlgo1(2,inf)}
    Algorithm~\ref{Algo1(2,inf)} terminates and outputs an allocation $\X$ that satisfies Properties (1), (2) and (4). It further preserves any of the Properties (3.1),(3.2) and (3.3), if they were satisfied before applying Algorithm~\ref{Algo1(2,inf)}.
\end{lemma}

\begin{proof}
    We first argue that Properties (1) and (2) are satisfied after any outer ``while'', i.e., the ones at lines 1-4 and 5-17, if they were satisfied before. Meanwhile, we show that any non-envied vertex cannot be envied after Algorithm~\ref{Algo1(2,inf)}. 
    
    Consider any round of the first ``while'' (lines 1-4), and suppose that Properties (1) and (2) were satisfied before this round. In the current round, a non-envied vertex $k$ gets a new bundle that it values at least as much as its previous bundle, which is composed of non-parallel bundles among the unallocated bundles and his own bundle (see Definition~\ref{def:UNPBnonenvied}). For each vertex $\ell$ adjacent to $k$, the new bundle of $k$ contains a single bundle relevant to $\ell$, that $\ell$ doesn't value more than its own bundle, due to Property (2) and to the fact that $k$ was non-envied. So, $k$ remains non-envied and Property (1) is preserved for the rest of the vertices. Regarding vertex $k$, by the definition of the $\UNPB_k(\X)$ set and the fact that the value of $k$ does not decrease in the new allocation, Property (1) is satisfied for $k$, as well. Property (2) is trivially satisfied for vertex $k$ by the definition of the $\UNPB_k(\X)$. Property (2) is also satisfied for the rest of the vertices since it was satisfied before the current round and due to the fact that $k$ was not envied (so, any bundles released by $k$ are not more valued to what the other vertices get). The same arguments hold for the ``while'' of lines 13-16, which is identical to the one in lines 1-4, if Properties (1) and (2) are satisfied before it is executed, as we will show next. 
    
    Regarding the ``while'' in lines 5-17, we will show that if Properties (1) and (2) were satisfied before each round of that ``while'', they are satisfied when the inner ``while'' at lines 10-12 terminates. In lines 5-9, an envied vertex $i$ receives a bundle that may be only envied by $p_i(\X)$, due to Property (2) and the definition of the set $\UNPB_i(\X)$ (see Definition~\ref{def:UNPB}). Moreover, Property (2) is satisfied for any other vertex apart from $p_i(\X)$, since the only bundle that was released was a bundle between $i$ and $p_i(\X)$ (see Observation~\ref{observation:orientation}). Therefore, the inner ``while'' in lines 10-12 will first run for $j=p_i(\X)$, where $p_i(\X)$ will receive the bundle that $i$ released and caused the envy towards $i$ in the first place. Overall, in lines 5-12 the number of envied vertices reduces by at least 1 (namely vertex $i$ becomes non-envied), and moreover in the ``while'' in lines 10-12 only an envied vertex may become non-envied and not vice versa, since a vertex may only change its assignment, after its neighbor releases a bundle while receiving a better one. Additionally, after the ``while'' in lines 10-12, Property (1) is satisfied since only relevant bundles are given to the vertices, their values may only increase and no more envy is introduced, and Property (2) is satisfied by the condition of the inner ``while'' (line 10).  

    Hence, we have showed that Properties (1) and (2) are satisfied at the end of Algorithm~\ref{Algo1(2,inf)}, and any initially non-envied vertex remains non-envied after Algorithm~\ref{Algo1(2,inf)}. The latter automatically means that if any of the Properties (3.1),(3.2) and (3.3) were satisfied before Algorithm~\ref{Algo1(2,inf)}, they are preserved after its termination.

    Regarding Property (4), consider first an envied vertex $i$ after the termination of Algorithm~\ref{Algo1(2,inf)}. It is trivial to see that Property (4) is satisfied for $i$, because $i$ does not satisfy the condition of line 5, otherwise Algorithm~\ref{Algo1(2,inf)} would not have terminated. We turn our attention to some non-envied vertex $k$ after the termination of Algorithm~\ref{Algo1(2,inf)}. Obviously $k$ does not satisfy the condition of line 13 (or the same condition of line 1, if the ``while'' in lines 5-17 has not been executed). This would mean that either i) $v_k(\UNPB_k(\X))< v_k(X_k)$, or 
    ii) $v_k(\UNPB_k(\X)) = v_k(X_k)$ and $|\UNPB_k(\X)| \leq |X_k|$. The former case is not possible due to the definition of $\UNPB_k(\X)$ that considers $X_k$ as a possible bundle. So, focusing on the latter case, due to the maximality of $\UNPB_k(\X)$, the set allocated to $k$ should include at least one bundle related to each neighbor (since the allocation is an orientation and it is not possible to have allocated both bundles to other vertices); we add this as an observation next to be used later. 

    \begin{observation}
    \label{obs:AllEdgesNon-enviedAllocatedINF}
        Let $\X$ be the allocation after Algorithm~\ref{Algo1(2,inf)}. For any non-envied vertex $k$, $X_k$ contains one bundle that is relevant to each of its neighbors. 
    \end{observation}

    Finally, we argue that Algorithm~\ref{Algo1(2,inf)} terminates. Note that at each round of the ``while'' at lines 1-4 or 13-16, either the social welfare (i.e., the aggregate value of all vertices) strictly increases, or for some agent $i$, the cardinality of $X_i$ with respect to its bundles from $B_i$ increases. So both ``while'' terminate (in pseudopolynomial time). The ``while'' in lines 10-12 strictly increases the social welfare, so for the same reason it terminates. The ``while'' in lines 5-17, runs at most $n$ times since at each round at least one vertex becomes non-envied and no non-envied vertex becomes envied. 
\end{proof}

\subsection{Step 3 - Final Allocation} At this final step, we start by an allocation satisfying Properties (1)-(4) (derived by Step 2) where Property (3) corresponds to one of the Properties (3.1),(3.2),(3.3) related to one of the three different cases, and we properly allocate the unallocated bundles, such that the complete allocation is EFX.

 We remark that the bundles allocated at Step 2, remain as they are, and we only allocate once and for all the unallocated bundles, to vertices other than the endpoints (so we only extend the bundles allocated in Step 2, in order to derive a full EFX allocation). The reason of allocating those remaining bundles to vertices that are not the endpoints is because if we would allocate them to one of the endpoints, this could possibly violate the EFX condition: e.g., if an envied vertex is allocated more edges, the envious vertex would not be EFX-satisfied anymore, and if a non-envied vertex receives more edges, those would be parallel to a bundle that he already has and may cause envy to one of his neighbors who even may not be EFX-satisfied anymore.  

Algorithm~\ref{Algo(FinalAllocationINF)} is used for deriving the complete EFX allocation. For this we define the set of unallocated bundles as   
    $UB(\X) = \bigcup_{i\in N} UB_i(\X)$, given an allocation $\X$. Algorithm~\ref{Algo(FinalAllocationINF)} does a simple thing: it allocates any unallocated bundle to a non-envied vertex that remains non-envied after this allocation. The existence of such non-envied vertices that are used for ``parking'' the unallocated bundles is guaranteed for each case separately by using the restriction of that case and Properties (3) and (4).
This is why we posed the graphical restrictions in the three cases, so that there is always a way to allocate those bundles without breaking EFX. Therefore, the unallocated bundles that are adjacent to an envied vertex $i$ are assigned to vertices that $i$ doesn't envy, even if they receive all the unallocated bundles adjacent to $i$ additionally to their bundle, under always the restriction of not containing parallel bundles (Property (4)). Similarly, the remaining unallocated bundles that are adjacent to some non-envied vertex $i$ are assigned to vertices that are not $i$'s neighbors and therefore $i$ has no value for their allocated bundle and any unallocated edges (Property (4)).

\begin{algorithm}
\caption{Complete EFX Allocation}
\label{Algo(FinalAllocationINF)}

\raggedright\textbf{Input:} An allocation $\X$ satisfying Properties (1)-(4).\\
\textbf{Output:} A complete EFX allocation.
\begin{algorithmic}[1]
\For{$\exists\; S \in UB(\X)$ with endpoints $i,j$} 
    \State Let $k$ be a non-envied vertex such that $X_k \cup S$ contains no parallel bundles, and  
         \Statex \hspace{25pt} for any $\ell\in\{i,j\}$, $v_{\ell}(X_{\ell}) \geq v_{\ell}\left(X_k \cup S\right)$ 
    \State $X_k \gets X_k\cup S$
\EndFor
\end{algorithmic}
\end{algorithm}

In the following lemma we show that Algorithm~\ref{Algo(FinalAllocationINF)} provides a complete EFX allocation for bipartite multigraphs and for multigraphs under the restrictions in the statements of Theorems~\ref{result2}~and~\ref{result2Girth}.

\begin{lemma}
For bipartite multigraphs, or for multigraphs where either there are at most $\lceil \frac{n}{4} \rceil - 1$ neighbors per vertex, or the shortest cycle with non-parallel edges has length at least 6, Algorithm~\ref{Algo(FinalAllocationINF)} returns a complete EFX allocation. 
\end{lemma}

\begin{proof}

Let $\X$ be the allocation from Step 2. This algorithm allocates each bundle to a vertex that is not an endpoint of the bundle (vertex $k$ in line 2). We show in the following claim that there is always the vertex $k$ of line 2, therefore Algorithm~\ref{Algo(FinalAllocationINF)} will terminate after allocating all unallocated bundles of $\X$. %
Then, we argue the allocation returned by  Algorithm~\ref{Algo(FinalAllocationINF)} preserves the EFX condition.

\begin{claim}\label{claim: safevertexINF}
For bipartite multigraphs, or for multigraphs where either the shortest cycle with non-parallel edges has length at least 6, or there are at most $\lceil \frac{n}{4} \rceil -1$ neighbors per vertex, there always exists the vertex $k$ of line 2 in Algorithm~\ref{Algo(FinalAllocationINF)}. %
\end{claim}

\begin{proof}
Note that if $r$ is the maximum number of bundles of any partition of parallel edges considered in the sets $B_i$, then the maximum number of unallocated bundles between two adjacent vertices $i,j$, where $q$ expresses the number of envied vertices in $\{i,j\}$, is at most 
\begin{equation}
    r+q-2, \mbox{ where } r=2 \mbox{ for Theorems~\ref{theorem:bipartite}~and~\ref{result2Girth}, and } r=3 \mbox{  for Theorem~\ref{result2}\,,} \label{ValueOfr}
\end{equation}

   This can be seen by Observation~\ref{obs:AllEdgesNon-enviedAllocatedINF}.
We give the following observation to be heavily used in the proof, and then we consider each of the three cases separately.
\begin{observation}
\label{obs:nonNeighbor-SafeINF}
    For any vertex $i$, if vertex $k$ is not $i$'s neighbor, then by Property (4), $v_i(X_i) \geq v_i(S) = v_i\left(S \cup X_k\right)$, for all $S\subseteq UB(\X)$ s.t. $S \cup X_k$ contains no parallel bundles.
\end{observation}

    {\bf Bipartite multigraphs.} Note that in this case by Property (3.1) there is no unallocated bundle between envied vertices and by \eqref{ValueOfr} there is no unallocated bundle between non-envied vertices. So, the only unallocated bundles are between an envied and a non-envied vertex. Consider any bundle $S \in U(\X)$ with endpoints $i,j$, such that $i$ is envied and $j$ is not envied. We will argue that $p_i(\X)$ is the required $k$ vertex: By Property (3.1) it holds that $p_i(\X)$ is non-envied, and  
    by Property (4), no envy can be created to $i$ by giving any unallocated bundles to $p_i(\X)$. Regarding vertex $j$, as an observation note that $j\neq p_i(\X)$, because one bundle between $i$ and $p_i(\X)$ has been allocated to $i$ and causing the envy of $p_i(\X)$, and the other has been allocated to $p_i(\X)$ due to Observation~\ref{obs:AllEdgesNon-enviedAllocatedINF}. Since $j$ is adjacent to $i$ and the graph is bipartite, $j$ is not adjacent to $p_i(\X)$, and by Observation~\ref{obs:nonNeighbor-SafeINF}, no envy can be created to $j$ by giving any unallocated edges to $p_i(\X)$.  

    {\bf At most $\lceil \frac{n}{4} \rceil -1$ neighbors.}

    We argue that for any pair of adjacent vertices $i,j$, there are $q+1$ (where $q$ is as defined in \eqref{ValueOfr}) non-envied vertices that are not adjacent to either $i$ or $j$. By \eqref{ValueOfr} and Observation~\ref{obs:nonNeighbor-SafeINF}, those are sufficient in order to prove the claim. The total number of $i$'s and $j$'s neighbors (including $i$ and $j$) is at most $2(\lceil \frac n4 \rceil - 1)\leq \lceil \frac n2 \rceil - 1$. So the number of vertices that are not adjacent to either $i$ or $j$ is at least $n-\lceil \frac n2 \rceil + 1$. Among them there are at most $\lfloor \frac n2 \rfloor - q$ envied vertices (by Property (3.2)), so the number of non-envied vertices that are not adjacent to either $i$ or $j$ is at least:
    $n-\lceil \frac n2 \rceil + 1 -(\lfloor \frac n2 \rfloor - q) \geq q + 1$.

    {\bf Length of shortest cycle with non-parallel edges at least 6.} This case is proven similarly to the setting with at most two parallel edges, with the only difference that instead of single edges we consider bundles. For the sake of completion we give the full proof here. Consider any bundle $S\in UB(\X)$ with endpoints $i,j$. By \eqref{ValueOfr}, at least one of the endpoints should be envied; so, w.l.o.g., let $i$ be envied. 
    
    If $j=p_i(\X)$ then $j$ must be envied: the reason is that $i$ is envied by $j$, so $i$ has received a bundle relevant to both $i,j$, and by Observation~\ref{obs:AllEdgesNon-enviedAllocatedINF}, $j$ has received the other bundle relevant to $i,j$, contradicting the fact that $S\in UB(\X)$. 
    Therefore, if $j=p_i(\X)$, then $j$ is envied, and there exists at most one unallocated bundle between $i$ and $j$; the other has been allocated to $i$. In that case, the role of $k$ in line 2 is given to either $p_j(\X)$ if it is non-envied, or to some non-envied neighbor of $p_j(\X)$ (which is guaranteed to exist by Property (3.3) when considering $j$). In both cases, $k$ is not $i$'s neighbor (since the shortest cycle with non-parallel edges has length at least 6). If $k\neq p_j(\X)$, then $k$ is not $j$'s neighbor either, and by Observation~\ref{obs:nonNeighbor-SafeINF}, $k$ satisfies the conditions of line 2. If $k=p_j(\X)$, then by also using Property (4) for $j$, $k$ satisfies the conditions of line 2.

    If $j\neq p_i(\X)$, by Property (3.3), either $p_i(\X)$ is non-envied or a neighbor of $p_i(\X)$ is non-envied. In the latter case, that neighbor is defined as $k$ which is not a neighbor of either $i$ or $j$, otherwise a cycle of length less than 6 would exist, which is a contradiction. By Observation~\ref{obs:nonNeighbor-SafeINF}, $k$ is a vertex satisfying the conditions of line 2. In the former case, we define $p_i(\X)$ as $k$, which is not a neighbor of $j$ for the same reason as above, and it holds $v_i(\UNPB_i(\X) \cup X_k)=v_i(\UNPB_i(\X))\geq  v_i(S \cup X_k)$, for all $S\in UB(\X)$ s.t. $S \cup X_k$ contains no parallel bundles. By Property (4) and Observation~\ref{obs:nonNeighbor-SafeINF}, $k$ is a vertex satisfying the conditions of line 2 if no other bundle parallel to $S$ has been given to $k$ before. If $j$ is non-envied, there is only one unallocated bundle between $i,j$ (due to \eqref{ValueOfr}), so the statement of the claim holds. We proceed the analysis with the case that $j$ is also envied and there may be two unallocated bundles between $i$ and $j$. Let $k_i$ be the $k$ defined above, i.e., the non-envied vertex guaranteed by Property (3.3), by considering the envied vertex $i$. Let also $k_j$ be similarly the non-envied vertex corresponding to vertex $j$; for the same reason as above, $k_j$ satisfies the conditions of line 2. Then, $k_i$ and $k_j$, that have distance at most 2 from $i$ and $j$, respectively, should be different, otherwise there would be a cycle of length at most 5. Then, since there are at most two unallocated bundles between $i,j$ (due to \eqref{ValueOfr}), one may be allocated to $k_i$ and the other to $k_j$, so the claim follows for that case, as well. 
\end{proof}

In Claim~\ref{claim: safevertexINF} we showed that for any unallocated bundle of $\X$ (the allocation of Step 2) there is always a non-envied vertex satisfying the conditions of line 2, and so Algorithm~\ref{Algo(FinalAllocationINF)} terminates in a complete allocation. The condition in line 2 guarantees that whenever an unallocated bundle is assigned to a vertex $k$, no new envy is caused towards $k$, and since $k$ is non-envied, EFX is preserved. 
\end{proof}

\section{Conclusion and Open Problems}
In this work we pushed the state-of-the-art regarding one of the most important and academically interesting problems in Fair Division: the existence of EFX allocations. Following the work of Christodoulou et al.~\cite{EFXsimplegraphs}, we consider a graph structure that poses a restriction to valuation functions. Based on the notation of \cite{EFXsimplegraphs}, we consider the $(2, \infty)$-bounded setting, which represents the case where each good is important for at most two agents. This can be seen as a multigraph setting where the vertices correspond to agents and the edges to goods, and each agent is interested only in goods/edges that are adjacent to her. Even if Christodoulou et al.~\cite{EFXsimplegraphs} showed the existence of EFX allocations in simple graphs, i.e., the $(2, 1)$-bounded setting, even for general monotone valuation functions, the generalization to multigraphs poses extra challenges: It is well known \cite{EFXsimplegraphs, OnTheStructureOfEFXOrientationsOnGraphs} that allocating edges to vertices other than their endpoints is inevitable, unless there are high restrictions on the graph structure. However, how to allocate those edges without violating EFX was the most challenging aspect in simple graphs settings \cite{EFXsimplegraphs}, and becomes much harder in the setting of multigraphs.

As a result, we pose some restrictions on the structure of the multigraphs in order to address the above challenges.  
We handle each restriction differently only with respect to some initial partial EFX orientation, where each agent receives a bundle that is relevant to only one of its neighbors. This initial allocation is carefully constructed so that it provides crucial properties that are essential in order to manage to assign edges to vertices that are different from their endpoints by preserving the EFX property, and as a result derive a complete EFX allocation. We remark that after the construction of the initial allocation, we follow a unified approach, which is a generalization of the approach in \cite{EFXsimplegraphs} to multigraphs. This further indicates that hopefully our approach and ideas will be useful to push the state-of-the-art even further.

One important future direction is to drop the structure restrictions on the multigraphs and show the existence of EFX allocations in multigraphs with general monotone valuations. Moreover, extending the existence of EFX allocations to hypergraphs, or other more general or incomparable $(p,q)$-bounded settings, could be a stepping stone towards the ultimate goal of showing EFX existence without any restriction on the setting; note that the setting of hyper-multigraphs corresponds to the unrestricted setting. 

At last, we remark that our approach heavily relies on the cut-and-choose protocol for general monotone valuation functions. It is known that finding such an allocation is hard (see \cite{PlautRough,GoldbergHH23}), so our algorithms are not efficient. Therefore, another aspect to consider is the complexity of finding such allocations, and it may be reasonable to consider approximate EFX allocations that can be computed efficiently. 

\section*{Acknowledgments}
    The research project is implemented in the framework of H.F.R.I call “Basic research Financing (Horizontal support of all Sciences)” under the National Recovery and Resilience Plan “Greece 2.0” funded by the European Union-NextGenerationEU (H.F.R.I. Project Number:15635).

\newpage

\bibliographystyle{plainnat}
\bibliography{refs}

\end{document}